\documentclass[english,journal]{IEEEtran}
\usepackage[T1]{fontenc}
\usepackage[latin9]{inputenc}
\usepackage{color}
\usepackage{amsmath}
\usepackage{amsthm}
\usepackage{amssymb}
\usepackage{stackrel}
\usepackage{graphicx}
\usepackage{multirow,multicol}

\makeatletter
\theoremstyle{definition}
\newtheorem{defn}{\protect\definitionname}
\theoremstyle{plain}
\newtheorem{thm}{\protect\theoremname}
\theoremstyle{plain}
\newtheorem{lem}{\protect\lemmaname}
\theoremstyle{plain}
\newtheorem{prop}{\protect\propositionname}
\theoremstyle{plain}
\newtheorem{cor}{\protect\corollaryname}

\@ifundefined{date}{}{\date{}}

\usepackage{amsfonts}
\usepackage{cite}
\usepackage{array}
\usepackage{algorithm}
\usepackage{algorithmic}
\usepackage{subfigure}
\usepackage{stfloats}
\usepackage{graphicx}

\usepackage[font=small,labelfont=bf]{caption}

\usepackage{gensymb}

\makeatother

\usepackage{babel}

\makeatother

\usepackage{babel}
\providecommand{\corollaryname}{Corollary}
\providecommand{\definitionname}{Definition}
\providecommand{\lemmaname}{Lemma}
\providecommand{\propositionname}{Proposition}
\providecommand{\theoremname}{Theorem}

\begin{document}
\title{Optimization-based Proof of Useful Work: Framework, Modeling, and
Security Analysis}
\author{Weihang~Cao,~\IEEEmembership{Graduate Student Member,~IEEE},
Xintong~Ling,~\IEEEmembership{Member,~IEEE}, \\
Jiaheng~Wang,~\IEEEmembership{Senior Member,~IEEE}, Xiqi~Gao,~\IEEEmembership{Fellow,~IEEE},
Zhi~Ding,~\IEEEmembership{Fellow,~IEEE} \thanks{The work of Xintong Ling, Jiaheng Wang, and Xiqi Gao was supported in part by the National Natural Science Foundation of China under Grants 62471137, U22B2006, and 72231002, the Natural Science Foundation on Frontier Leading Technology Basic Research Project of Jiangsu under Grants BK20222001 and BK20212001, the Jiangsu Province Major Science and Technology Project under Grant BG2024005, the Science and Technology Major Project of Nanjing under Grant 202405020, the Fundamental Research Funds for the Central Universities under Grant 2242022K60002, and the Young Elite Scientists Sponsorship Program By CAST 2023QNRC001. \textit{(Corresponding Author: Xintong Ling.)}\\
Weihang Cao is with the National Mobile Communications Research Laboratory,
Southeast University, Nanjing 211189, China (e-mail: whcao@seu.edu.cn).
\protect \\
Xintong Ling, Jiaheng Wang, and Xiqi Gao are with the National Mobile
Communications Research Laboratory, Southeast University, Nanjing
211189, China, also with the Purple Mountain Laboratories, Nanjing
210023, China (e-mail: xtling@seu.edu.cn; jhwang@seu.edu.cn; xqgao@seu.edu.cn).\protect \\
Zhi Ding is with the Department of Electrical and Computer Engineering,
University of California at Davis, Davis, CA 95616 USA (e-mail: zding@ucdavis.edu).\textit{ }}  \vspace{-0.8cm}}
\maketitle
\begin{abstract}
Proof of Work (PoW) has extensively served as the foundation of blockchain's
security, consistency, and tamper-resistance, but long has it
been criticized for its tremendous and inefficient utilization of
computational power and energy. Proof of useful work (PoUW) can effectively
address the blockchain's sustainability issue by redirecting the computing
power towards useful tasks instead of meaningless hash puzzles. Optimization
problems, whose solutions are often hard to find but easy to verify,
present a viable class of useful work for PoUW. However, most existing
studies rely on either specific problems or particular algorithms,
and there lacks comprehensive security analysis for optimization-based
PoUW. Therefore, in this work, we build a generic PoUW framework that
solves useful optimization problems for blockchain consensus. Through
modeling and analysis, we identify the security conditions against
both selfishness and maliciousness. Based on these conditions, we establish
a lower bound for the security overhead and uncover the trade-off
between useful work efficiency and PoW safeguard. We further offer
the reward function design guidelines to guarantee miners' integrity.
We also show that the optimization-based PoUW is secure in the presence
of malicious miners and derive a necessary condition against long-range
attacks. Finally, simulation results are presented to validate our
analytical results.
\end{abstract}


\section{Introduction }

 Blockchain has emerged as a disruptive technology that leads revolutionary
advances in multiple fields \cite{Ling2019,Chen2025,Ling2020c,Jiang2021}.
Blockchain assures integrity and reaches agreement among untrusted
entities, while maintaining high data security, consistency, and auditability \cite{Wu2026,Ling2025}.
The consensus protocol employed by a blockchain system lies at the
core of its functionality. Among the most representative consensus
protocols, proof of work (PoW) has been widely adopted in various
blockchain systems such as Bitcoin \cite{Nakamoto2008}, due to its
well-proven security and scalability performance. However, PoW has
been extensively criticized for its tremendous energy consumption
\cite{Wang2019}. In PoW, blockchain maintainers, or miners, generate
blocks by brute-forcing a cryptographic puzzle that is difficult to
solve but easy to verify. Block generation, also known as mining,
is thus highly computation-demanding and consumes a considerable amount
of computational power and electricity. Furthermore, the PoW puzzles
hold no practical meaning beyond maintaining blockchain consensus.
Therefore, the energy consumption in blockchain systems remains a
pressing challenge. 

A promising approach to address the power consumption issue is to
replace the useless PoW puzzles with some practical real-world puzzles
so that the invested computational power can be repurposed for useful
works. This type of consensus mechanism is generally referred to as
proof of useful work (PoUW) \cite{Ball2017}. The useful works explored
by early studies include DNA sequence alignment \cite{Ileri2016},
traveling salesman problems (TSP) \cite{Loe2018}, discrete logarithms
\cite{Hastings2018}, matrix multiplications \cite{Shoker2017}, etc.
However, some of them are not based on user-uploaded problems but
randomly generated problems, whose usefulness is thus questioned \cite{Loe2018,Hastings2018}.
Some others rely on trusted third parties or trusted hardware \cite{Ileri2016,Shoker2017},
which greatly jeopardizes the decentralization of blockchain and induces
more security risks.     

In addition to usefulness, a robust PoUW should also satisfy several
fundamental properties, including efficiency, completeness, soundness,
and hardness \cite{Ball2017}. Note that optimization problems are
aligned with these features. Combining optimization with consensus
mechanisms has thus been explored in several works. The work \cite{Ball2017}
proposed useful proof of work (uPoW) based on orthogonal vectors (OV)
problems. The authors of \cite{Oliver2017,Philippopoulos2020} proposed
a difficulty-based incentive for problem solving (DIPS) scheme by
reducing the difficulty constraint of PoW if the block contains a
solution better than the current best solution. Hybrid mining \cite{Chatterjee2019}
allows miners to add new blocks by either solving a submitted problem
or through traditional PoW hash trials, and problem-solving-based
blocks are announced in a two-step manner to protect the solutions.
Shibata \cite{Shibata2019} proposed proof-of-search that allows computational
power for making a consensus to be used for finding good approximate
solutions for optimization problems. Proof-of-search separates best
solution determination and block generation, and uses a modified hash-based
calculation process to find an intermediate solution that functions
similarly to the nonce field in PoW. However, in proof-of-search,
optimization problems are treated as black boxes whose structures
are unknown to miners, so that miners can only randomly guess an approximate
solution repeatedly, which significantly reduces computation efficiency.
Inspired by proof-of-search, \cite{Bizzaro2020} proposed proof of
evolution (PoE) by specifying the optimization method as a genetic
algorithm. In \cite{Davidovic2022,Todorovic2022}, a combinatorial
optimization consensus protocol (COCP) that solves combinatorial optimization
(CO) problems is proposed. In COCP, optimization instances are posted
by customers to form an instance pool, solved by the miners to create
blocks, and approved by verifiers to be added to the blockchain. However,
a malicious miner can directly copy other miners' solutions to compete
with the original solution, posing a significant threat to blockchain
security. In \cite{Fitzi2022}, a PoUW-based consensus protocol, Ofelimos,
is proposed built on a doubly parallel local search (DPLS) algorithm,
which repeatedly explores the neighborhood of a currently visited
location in the solution space to obtain an improved solution. Ofelimos
incorporates a pre-hash method that prevents grinding the search point,
together with a post-hash method that determines the validity of a
block similar to PoW. However, Ofelimos does not allow miners to choose
the initial point freely, which limits the optimization performance.
Recently, another important line of work is to reach consensus by
training machine learning (ML) models as an alternative to PoW puzzles
\cite{Baldominos2019,BravoMarquez2019,Chenli2019,Liu2021,Qu2021,Wang2022,Wei2023,Xu2024}.
Essentially, ML problems fall into the category of optimization, except
that they require an extra verification process based on some test
datasets.         

Despite various efforts, existing solutions still suffer from multiple
limits. First, most proposed consensus protocols are based on specific
problems, particular algorithms or even hardware, thereby limiting
their applications. For example, the problem is specified as TSP problems
in \cite{Loe2018}, and in Ofelimos \cite{Fitzi2022} the search algorithm
is specified. For another example, the paper \cite{Zhang2017} proposed
resource-efficient mining (REM) by employing Intel software guard
extensions (SGX) as the trusted executable environment (TEE). Though
the assumption of TEE to execute and remotely verify software running
assures low overhead in a partially decentralized network, its heavy
dependence on trusted hardware and centralized verification (which
is managed by Intel) leads to potential risks and limits its applications,
especially in permissionless blockchains. Second, most works lack
rigorous security analysis, which is essential for optimization-based
PoUW. Most papers only present their solutions and argue that they
can work without any theoretical analysis \cite{Oliver2017,Philippopoulos2020,Chatterjee2019,Shibata2019,Bizzaro2020,Davidovic2022,Todorovic2022},
making their security performance questionable. Third, it is important
to emphasize that optimization-based PoUW needs to consider additional
security issues such as solution plagiarism, which is often omitted
in existing papers. If solutions in existing blocks can be copied
and used to generate an alternative block easily, the security of
the blockchain would be greatly threatened. Though some papers discussed
solution plagiarism, their methods are either too complicated or require
specific algorithms, which restricts the optimization performance.
For instance, some papers use two steps to reveal the solutions \cite{Chatterjee2019,Chen2022},
i.e., miners upload the digest of their solutions in step one, and
reveal their solutions in step two to determine the winner. However,
this method has high synchronization requirements and is thus unsuitable
for distributed networks. Another method is to insert identity-related
puzzles in the search algorithms \cite{Bizzaro2020,Fitzi2022}. However,
this may lead to low optimization efficiency. It therefore calls for
a more generalized framework of optimization-based PoUW and the corresponding
theoretical characterization.  

To address the above challenges, this study formalizes a generic optimization-based
PoUW framework that integrates blockchain consensus with real-life
optimization problems, which is representative of a broad range of
optimization-based consensus protocols. Our framework serves dual
functions of solving useful optimization problems and maintaining
blockchain consensus simultaneously. We design a simple and intuitive
workflow to achieve these two objectives. Unlike existing works, our
framework imposes minimal requirements on the optimization problem
formulations and is not restricted to specific solving methods. Therefore,
our framework is generic and applicable to a diverse range of optimization
problems. We model the above general framework and conduct a thorough
security analysis for both selfishness and maliciousness. Specifically,
the major contributions of our work are summarized as follows.  
  
\begin{itemize}
\item We formalize a generic optimization-based PoUW framework that maintains blockchain consensus by solving optimization problems. Notably, our framework is an abstract of a wide category of optimization-based PoUWs.
\item We summarize the critical threats of optimization-based PoUW, establish the corresponding analytical model, and identify the security conditions against selfishness and maliciousness, respectively. 
\item We derive a lower bound for the security overhead and discuss the
trade-off between useful work efficiency and PoW safeguard. 
\item We investigate and uncover the reward function design principle, which
is equivalent to the security conditions against selfishness, and
further provide a linear reward design. 
\item We demonstrate that our framework is secure against malicious miners,
and derive a necessary condition against long-range attacks.
\item We conduct comprehensive simulations to demonstrate the performance
of our framework. Simulation results verify the security boundaries
derived from our theoretical analysis.
\end{itemize}
\par The remainder of the paper is organized as follows. Section
\ref{sec:Dual-functional-Consensus-Framew} demonstrates the framework,
and Section \ref{sec:Modeling} provides the system model. In Section
\ref{sec:Model-Analysis}, we analyze selfish miners and derive the
security conditions against selfishness. We show the security overhead
lower bound in Section \ref{sec:PoW-Safeguard-Versus}, and provide
the reward function design principle in Section \ref{sec:Reward-Function-Design}.
We further derive the security condition
against maliciousness in Section \ref{sec:Long-range-Attack-Analysis}.
The simulation is presented in Section \ref{sec:Simulations-and-Discussions}.
Finally, we conclude our work in Section \ref{sec:Conclusions}.

\section{Optimization-based PoUW \label{sec:Dual-functional-Consensus-Framew}}

\subsection{Overall Framework}

In this section, we formalize a generic optimization-based PoUW framework.
We consider a permissionless and anonymous blockchain network with
three types of participants: requesters, miners, and blockchain users.
Requesters possess optimization problems, referred to as tasks, and
are willing to remunerate the miners for solving them. Miners utilize
their computational power not only to solve the tasks, but also to
maintain the blockchain network through block generation and validation.
Blockchain users pay transaction fees or gas fees to miners for recording
their transactions on-chain or executing smart contracts. 

The optimization tasks uploaded by the requesters need to be defined
in the following form:
\begin{align}
\text{maximize} & \quad y=f\left(x\right)\\
\text{subject to} & \quad x\in\chi,\nonumber 
\end{align}
where $x$ is the optimization variable, $y=f\left(x\right)$ is the
optimization objective, and $\chi$ is the feasible set of the problem.
To initiate a task, the requester is required to lock some of its
credit in advance as the bounty for miners who provide good solutions
to the task. Each task contains a benchmark solution $\left(x^{(0)},y^{(0)}\right)$,
which serves as a baseline for the solutions submitted by miners.\textbf{
}
\begin{figure}
\begin{raggedright} \centering\includegraphics[width=0.9\columnwidth]{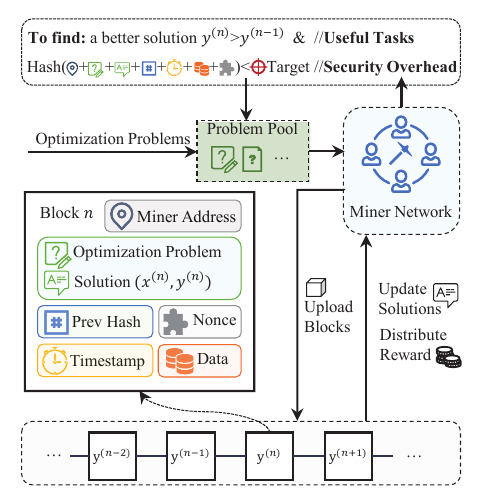}\end{raggedright}
\centering{}

\caption{A generic framework of optimization-based PoUW blockchain.\label{fig:architecture}}
\vspace{-0.5cm}
\end{figure}

The workflow of the optimization-based PoUW is illustrated in Fig.
\ref{fig:architecture}. The tasks generated by requesters form a
problem pool, similar to the transaction pool in cryptocurrencies.
The newly-arrived tasks are then organized into blocks and added to
the blockchain by miners. Miners solve the pending optimization tasks
to generate new blocks. When a miner successfully obtains a better
solution $\left(x^{(1)},y^{(1)}\right)$ than the current benchmark,
it integrates all the elements including the miner's address $addr$,
a reference to the optimization problem $op$, the better solution
$\left(x^{(1)},y^{(1)}\right)$, the hash value of the previous block
$h$, the current timestamp $ts$, a randomly chosen number called
the nonce, and the other payload data such as transactions and smart
contracts, to form a new potential block. After obtaining a better
solution $\left(x^{(1)},y^{(1)}\right)$, the miner needs to repeat
the above process until it finds a suitable nonce that makes the block
hash smaller than a certain threshold $T$, i.e.,
\begin{equation}
\text{Hash}\left(addr||op||\left(x^{(1)},y^{(1)}\right)||h||ts||\text{nonce}||\text{data}\right)<T.\label{eq:hash_trial}
\end{equation}
After finding a valid solution and a suitable nonce, the miner packages
all the information into a block and sends it to the blockchain network.
The hash trials in \eqref{eq:hash_trial} serve as the security overhead
that helps prevent the original solution from being plagiarized. Because
of \eqref{eq:hash_trial}, other miners cannot directly copy the solution
and claim a block of their own, but must perform hash trials until
a new nonce is found, as the input becomes different. It somehow provides
the original block with a certain temporal advantage and safeguards
the honest miners' interest. 

Upon receiving the mined block, other miners first validate the correctness
of the solution through the following steps:
\begin{enumerate}
\item Correctness: check whether $y^{(1)}=f\left(x^{(1)}\right)$.
\item Feasibility: determine whether $x^{(1)}$ satisfies the constraints,
i.e., whether $x^{(1)}\in\chi$.
\item Superiority: check whether $y^{(1)}>y^{(0)}$, i.e., the solution
is better.
\end{enumerate}
The time for solution verification is trivial compared with finding
a solution. Miners then confirm whether the block hash satisfies the
difficulty requirement. The newly received block is accepted if both
the solution and the hash puzzle are valid. Then, the benchmark of
the corresponding task is updated with the new solution, so the next
time the same task is employed for block generation, a solution better
than the new benchmark $y^{(1)}$ is required. After block validation,
other miners can update their current best solution with the new benchmark,
provided that it is superior, and proceed with optimization-solving
based on the updated solution. By this means, the solution keeps evolving
towards a better one. 

If a superior solution cannot be found, miners can generate a block
by PoW but with a much higher difficulty. The expected time to create
such a PoW block is much longer than that required to generate an
optimization-based block in the worst case, to encourage useful-work-based
mining when better solutions can be found. A task is automatically
terminated when a pure PoW block is mined, and the bounty will be
distributed to all the miners who have committed to the improvement
of the solution. After that, miners will switch to the next optimization
problem.   

Compared with existing consensus mechanisms relying on optimization
puzzles, our framework accepts general optimization problems as useful
work, while most existing works can only solve a specific type of
problem. By incorporating PoW as a security overhead, our framework
inherits the security features from PoW, which is more secure than
many existing approaches, especially those that fail to consider solution
plagiarism. However, security comes at the price of efficiency in
solving useful works, which will be discussed in Section \ref{sec:PoW-Safeguard-Versus}.

\subsection{Security Implications \label{subsec:2.2}}

Optimization-based PoUWs yield more security implications. This subsection
would like to first identify two primary selfish misbehavior in the
context of optimization-based PoUW. 

\textbf{Solution plagiarism}. Optimization-based PoUWs encounter an
inherent contradiction: the necessity to disclose solutions for verification
and fostering competitive cooperation exposes them to the risk of
plagiarism by dishonest miners. In our framework, the calculation
of PoW grants the miner a time advantage that helps secure the solution.
When a miner receives a new block and wants to duplicate the solution
directly to create a fork, it still needs to perform PoW hash trials
to claim block rewards, as the address in \eqref{eq:hash_trial} is
different. The PoW computation introduces certain latency, during
which the honest block has already been propagated to the entire network.
Despite the latency introduced by PoW, it remains insufficient to
completely prevent solution plagiarism, which requires in-depth mathematical
characterization. 

\textbf{Chain formation}. The solution quality is involved in optimization-based
PoUW, which influences miners' choices between forming a chain and
deliberately forking for higher profits. Intuitively, miners may prefer
linking to a preceding block with a poor solution, since a greater
improvement relative to a worse benchmark increases the probability
of finding a valid solution with potentially higher block reward.
If it is profitable to deliberately create forks, then PoUW will fail
to form a chain and cannot function properly, which should be carefully
addressed.

Apart from selfish misbehavior, malicious attacks also need to be
considered. We identify two typical malicious attacks in optimization-based
PoUW. 

\textbf{Tampering with the chain history.} To tamper with the blockchain
by the majority attack poses a fundamental challenge to all blockchains.
In our framework, after acquiring an improved optimization solution,
miners must find appropriate nonces to meet the PoW requirement. Therefore,
the framework naturally inherits the secure properties of PoW, necessitating
an honest majority. In the absence of a benign majority, there exists
the vulnerability of a powerful attacker manipulating the blockchain
data.

\textbf{Long-range Attack.} Though PoW helps secure original blocks
by inducing delays for plagiarists, it cannot directly prevent the
solutions from being stolen. If a malicious miner applies a long-range
attack, i.e., it tries to create a longer fork from deep in the chain,
it can directly copy the solutions for creating its private chain
and only solves the PoW part of the puzzles. This gains the attacker
an advantage in catching up with the main chain, making the honest
majority assumption insufficient. 

Note that, even the pure PoW protocol cannot completely avoid these threats like alternative history attacks \cite{Ling2020c}. Therefore, it is imperative to give the security conditions that can guarantee the system safety under these attacks. We will derive the corresponding results via rigorous modeling and mathematical analysis in the ensuing sections.

\section{PoUW Blockchain Modeling \label{sec:Modeling}}

\subsection{Network Model }

In this work, we apply a discrete-time network model \cite{Garay2015},
in which miners take actions stepwise according to fixed time intervals
called rounds. In one round, the number of hash trials of the whole
blockchain network is denoted by $q_{\text{total}}$, representing
the total computing power of the network. According to \eqref{eq:hash_trial},
a hash trial converts the elements in a block into a binary string
of length $L$ (e.g., $L=256$ in Bitcoin blockchain). Due to the
properties of the hash function, hash trials can be modeled as a sequence
of independent Bernoulli trials with success probability $\frac{T}{2^{L}}$,
where $T$ is the difficulty target as defined in Section II-A. Let
$q_{0}$ be the probability that the entire miner network obtains
a qualified hash value within a given round, provided that all computational
power is devoted to hash trials. Usually $\frac{T}{2^{L}}\ll1$, and
roughly, we have $q_{0}\approx\frac{q_{\text{total}}T}{2^{L}}$. It
is evident that $q_{0}$ is influenced by two key factors: the total
amount of computational power in the network and the threshold $T$. 

Similarly, we assume that in each round, for a given optimization
problem, the probability of obtaining a better solution is denoted
by $p_{0}$. Note that $p_{0}$ is determined by the computational
power of the network invested in optimization problem-solving, as
well as the inherent hardness of the optimization problem itself.
It is a general assumption that can be applied to various kinds of
optimization problems.   
\begin{figure*}
\begin{raggedright} \centering\includegraphics[width=0.88\textwidth]{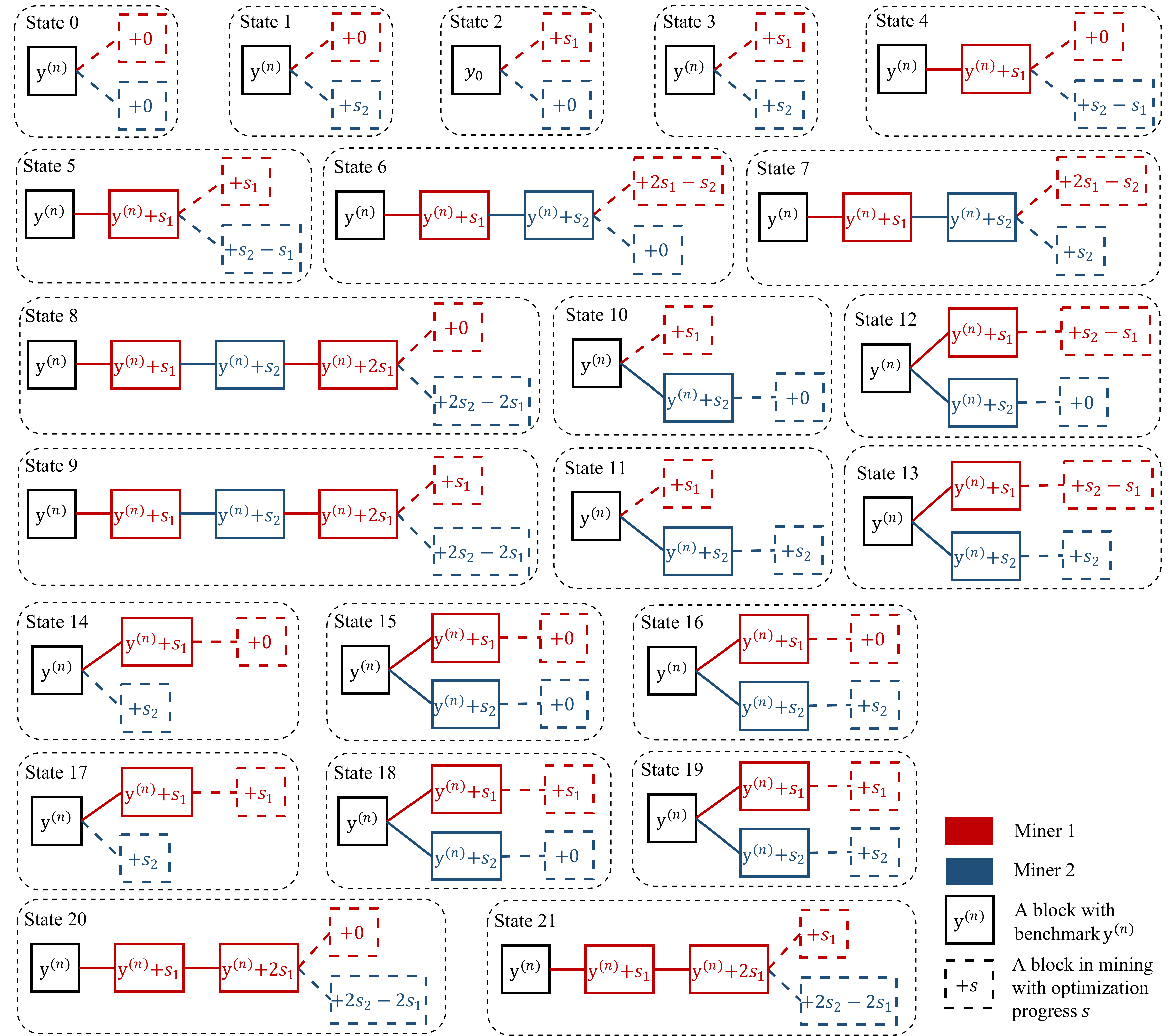}\end{raggedright}
\centering{}

\caption{Illustration of states for different strategy profiles.\label{fig:states}}
\vspace{-0.5cm}
\end{figure*}

\subsection{Honest Miner Model}

Now consider the miner $i$ that behaves honestly according to the
framework. We denote miner $i$'s ratio of computational power in
the blockchain network as $\lambda_{i}$. Specifically, in each round,
miner $i$ applies the following honest (H) strategy to generate a
new block:
\begin{itemize}
\item Denote the current best-known solution (benchmark) recorded in the
blockchain as $y^{(n)}$. Miner $i$ will keep optimizing the problem
until it finds a solution $y^{(n+1)}$ better than the benchmark $y^{(n)}$.
Denote the success probability of finding a better solution in each
round as $p_{i}$, and we assume that $p_{i}$ is proportional to
miner $i$'s computational power $\lambda_{i}$, i.e., $p_{i}=\lambda_{i}p_{0}$. 
\item According to our optimization-based PoUW framework, if miner $i$'s
current solution surpasses $y^{(n)}$, it then performs hash trials
to find a nonce that meets the threshold requirement. The success
probability of finding a suitable nonce in a round is $q_{i}$, and
obviously, $q_{i}=\lambda_{i}q_{0}$. In each round, a miner may attempt
to find a better solution than the current benchmark, or perform hash
trials if it has obtained a valid solution. (We assume that these
two options are exclusive in one round because the computational power
of each miner is limited and the round duration is short.)
\item After miner $i$ generates a valid block with a new solution $y^{(n+1)},$
it will obtain the block reward $R\left(s\right)$ according to the
reward function $R(\cdot)$, where $s$ represents the improvement
of the new solution $y^{(n+1)}$, compared with the previous benchmark
$y^{(n)}$. Apparently, the mining strategy is significantly influenced
by the reward function $R(\cdot)$, which thus should be designed
carefully and defined with the optimization task. 
\item The value of $s$ implies the effort of the miner devoted to finding
the new solution $y^{(n+1)}$, given the previous benchmark $y^{(n)}$.
The definition of $s$ can vary. One common approach is to directly
quantify the improvement as $s=y^{(n+1)}-y^{(n)}$. However, considering
the characteristics of most optimization algorithms, we can define
$s$ as the order of convergence, satisfying $\underset{n\rightarrow\infty}{\lim}\frac{|y^{(n+1)}-\bar{y}|}{|y^{(n)}-\bar{y}|^{s}}=\text{constant}$,
where $\bar{y}$ is the limiting point of the optimization algorithm.
Linear, quadratic, and superlinear convergence patterns correspond
to $s=1$, $s=2$, or $1<s<2$, respectively. In practice, we can
calculate the order of convergence $s$ from previous solutions: $s\approx\log\left(\frac{y^{(n+1)}-y^{(n)}}{y^{(n)}-y^{(n-1)}}\right)/\log\left(\frac{y^{(n)}-y^{(n-1)}}{y^{(n-1)}-y^{(n-2)}}\right)$
\cite{Senning2019}. In the following analysis, we assume that the
optimization capability $s$ for miner $i$ is nearly constant for
an optimization task. 
\item If another miner generates a new block and miner $i$ has not mined
one yet, then it would update its local chain together with the current
benchmark $y^{(n)}$. If the new benchmark surpasses the miner's current
solution, it will update its current best-known solution and attempt
to calculate a better one. Otherwise, it will carry out hash trials
after the new block in an attempt to generate the subsequent block. 
\end{itemize}

\subsection{Adversary Miner Model}

Two types of adversary miners need to be considered: the first type
is selfish but rational miners who may deviate from the protocol for
their interests; the second type is purely malicious miners who will
attack the blockchain even at their losses. The analysis of selfish
miners is particularly important in a PoUW blockchain, as the selfish
miners may intentionally fork and benefit from claiming the block
rewards. We first consider selfish miners, and in Section \ref{sec:Long-range-Attack-Analysis},
we will discuss malicious miners. 

Assume all miners apply the longest chain rule. That is, when the
miners' views differ (which, in our model, only occurs when a selfish
miner intentionally forks), if their chains are of the same length,
they will continue to mine on their own chain. If the selfish miner
creates a fork that is one block longer than the honest miners' chains,
then the honest miners will abandon their chain and adopt the longer
fork. Conversely, the selfish miner will give up forking if the honest
miners' chain becomes longer. We ignore more complex cases in which
the selfish miner competes for a longer time.  

Under these assumptions, it can be observed that the optimization
progress of miners can affect the dishonest strategies the selfish
miner may deploy. Consider a mining competition between two miners,
miner 1 and miner 2. We assume their optimization improvement $s_{1}<s_{2}$
without loss of generality.

When the selfish miner has a smaller optimization improvement, i.e.,
miner 1 is selfish, it may apply a strategy named fork-and-steal (FS)
as follows. When a new block is generated by the honest miner before
the selfish miner does, if the selfish miner has already acquired
a solution $y^{(n)}+s_{1}$, which is better than the previous benchmark
$y^{(n)}$ but smaller than the honest miner's new solution $y^{(n)}+s_{2}$,
instead of abandoning its solution and replacing it with the new best
solution, the selfish miner will do the following: it secretly mines
a block with solution $y^{(n)}+s_{1}$, then steals the solution $y^{(n)}+s_{2}$
and mines another block. If the selfish miner succeeds before the
honest miner generates the next block, it will announce the fork,
which is one block longer than the honest miner's single block. Due
to the longest chain, this fork is accepted by the main chain and
the honest miner's new block will be replaced. If the honest miner
generates the next block before the selfish miner successfully creates
two blocks, the forking attempt is considered failed, and the selfish
miner will update its solution and restart honest mining. 

Also, when the selfish miner has a more significant optimization progress,
i.e., miner 2 is selfish, a strategy named ignore-and-fork (IF) may
be utilized. When the honest miner generates a block with solution
$y^{(n)}+s_{1}$, if the selfish miner already has solution $y^{(n)}+s_{2}$,
which is better, it may secretly mine a fork parallel to the honest
miner's newly found block. As long as the selfish miner can generate
two blocks, i.e., with solution $y^{(n)}+s_{2}$ and $y^{(n)}+2s_{2}$,
respectively, before the honest miner generates the next block, then
it will make its secret chain public, which will be accepted as the
main chain since it is one block longer. If it fails, the selfish
miner will return to honest mining. 

Notably, the FS and IF strategies correspond to the solution plagiarism
and the chain formation misbehavior, respectively, as discussed in
Section \ref{subsec:2.2}. These two selfish strategies are typical,
as they only start acting selfishly when the selfish miner already
has some advantage, i.e., has an unused solution. Therefore, these
two strategies are the primary concern of our security analysis against
selfishness.

A blockchain network has massive miners and might have multiple selfish
miners that attempt to maximize their own interests. Now we use a
stronger assumption that these selfish miners are colluded so that
the whole blockchain network can be abstracted into two parties: the
honest majority and the selfish rest. The honest one always abides
by the designed framework, whereas the selfish one may not follow
the protocol for a better payoff. We denote the ratios of the computational
power of these two parties as $\lambda_{1}$ and $\lambda_{2}$, respectively,
and obviously $\lambda_{1}+\lambda_{2}=1$. We also assume that the
optimization improvements of these two parties are on par with each
other, such that $s_{1}<s_{2}<2s_{1}$, which implies that the optimization
algorithms applied by both parties do not differ too much in performance.

The above modeling serves as a foundation for the analysis of selfishness.
Intuitively, if deviating from honest mining lowers mining rewards,
then honest mining would be a stable strategy in the considered setup.
It guarantees rational miners follow the framework and behave honestly,
thus a prerequisite for a secure optimization-based PoUW. The next
section will investigate the average rewards for different strategies.
\begin{figure}
\begin{raggedright} \centering\subfigure[]{ \includegraphics[width=0.95\columnwidth]{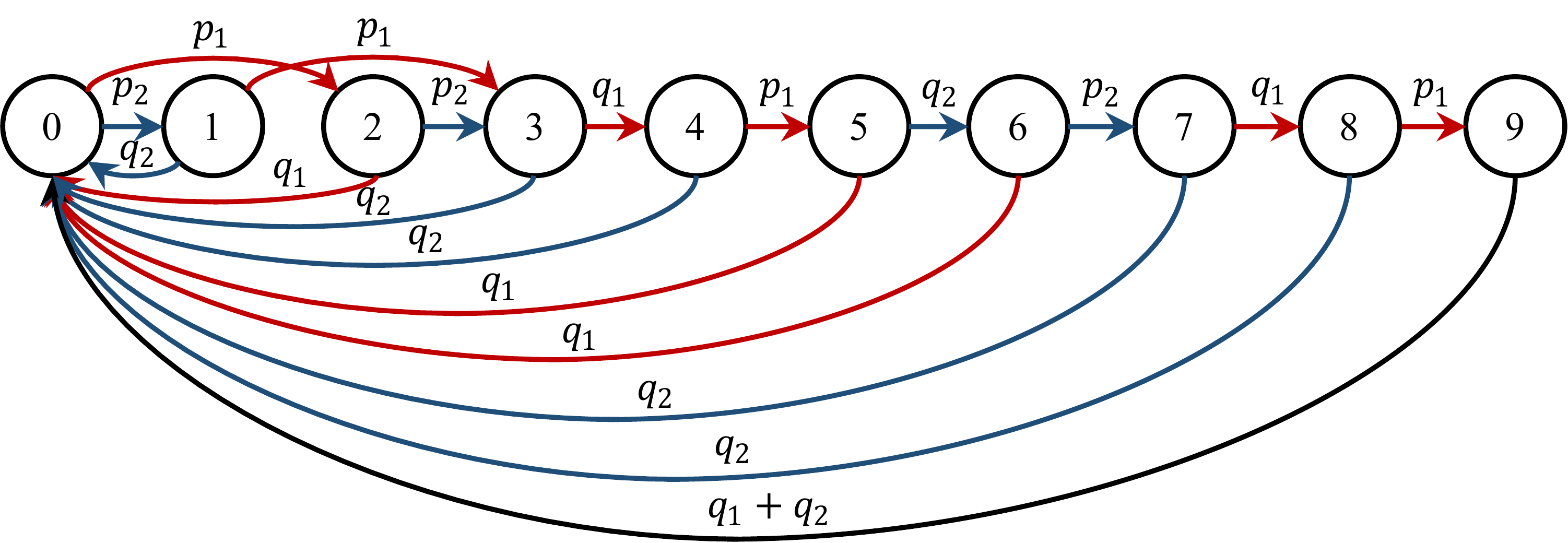}}
\hfill{}\subfigure[]{ \includegraphics[width=0.95\columnwidth]{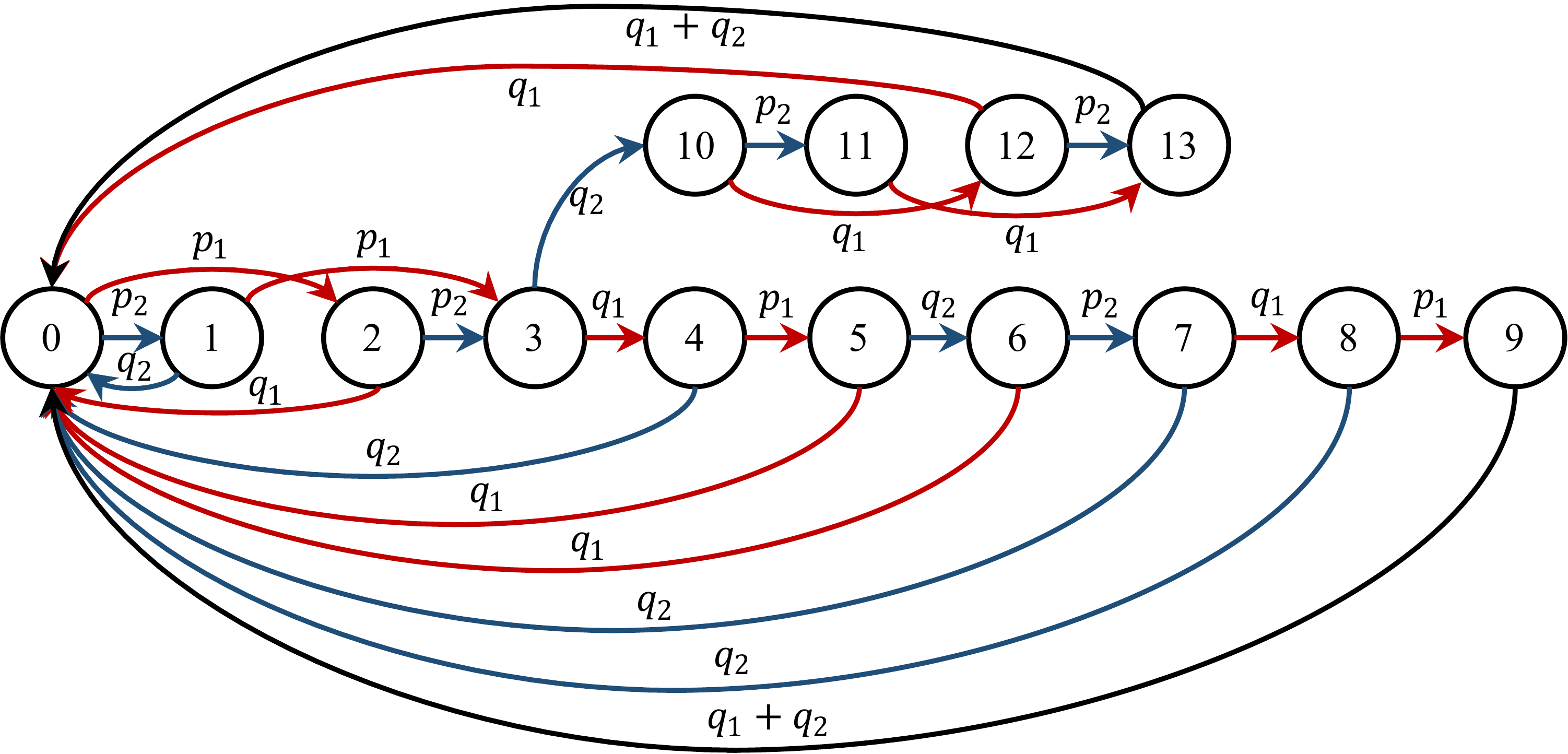}}
\hfill{}\subfigure[]{ \includegraphics[width=0.95\columnwidth]{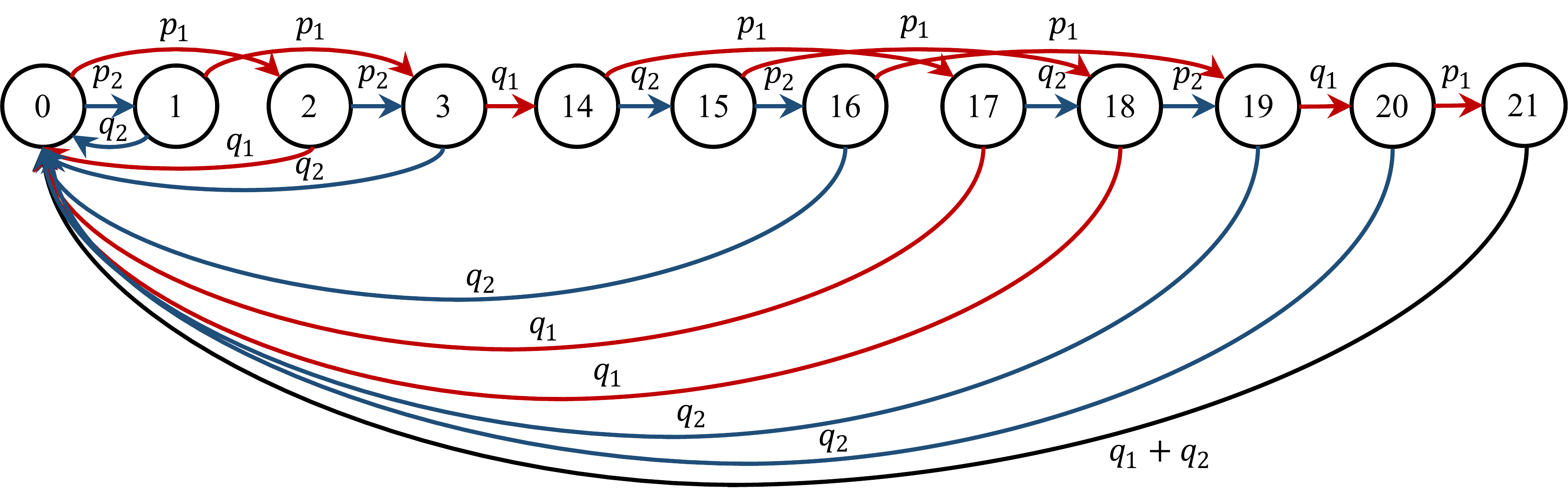}}

\end{raggedright} \centering {}\vspace{-0.2cm}

\caption{The state transition diagram for different strategy profiles. The
numbers in circles represent the corresponding states shown in Fig.
\ref{fig:states}. (a) (H, H). (b) (FS, H). (c) (H, IF). \label{fig:queue}}
\vspace{-0.5cm}
\end{figure}

\section{Security Against Selfishness \label{sec:Model-Analysis}}

In this section, we derive the payoff for different strategy profiles
and explore the conditions of guaranteeing honest mining. Based on
the above model, we investigate whether the strategy profile (H, H)
in which both miners adopt the honest strategy constitutes a Nash
equilibrium, where any deviation from the current state decreases
the miner's expected payoff. If (H, H) is a Nash equilibrium, both
miners are expected to act honestly to maximize their overall gain,
thereby ensuring the security of optimization-based PoUW. Therefore,
we aim to identify the security conditions against selfishness, as
shown in Definition \ref{def:1}. 
\begin{defn}
$\pi_{i}\left(\mathcal{\sigma}_{1},\mathcal{\sigma}_{2}\right)$ is
denoted as the expected payoff per round for miner $i$ ($i=1,2$),
where $\mathcal{\sigma}_{1}\in\left\{ \text{H},\text{FS}\right\} $
and $\mathcal{\sigma}_{2}\in\left\{ \text{H},\text{IF}\right\} $. 
\end{defn}
\begin{defn}[Security Conditions against Selfishness]
\label{def:1}  The strategy profile (H, H) is a Nash equilibrium
under the following constraints:
\begin{align}
\pi_{1}\left(\text{H},\text{H}\right) & >\pi_{1}\left(\text{FS},\text{H}\right),\label{eq:security_cond1}\\
\pi_{2}\left(\text{H},\text{H}\right) & >\pi_{2}\left(\text{H},\text{IF}\right).\label{eq:security_cond2}
\end{align}
\eqref{eq:security_cond1} and \eqref{eq:security_cond2} are thus
defined as security conditions against selfishness, which can guarantee
that honest mining yields a higher profit than selfish strategies
if others behave honestly.
\end{defn}
Definition \ref{def:1} is the basis of the subsequent analysis. In
the following subsections, we model the state transition of the different
strategy profiles ((H, H), (FS, H) and (H, IF)) and derive the expected
payoffs. 

\subsection{Honest Strategy}

We apply discrete Markov chains, in which the transition probability
only depends on the previous state, and at most one state transition
can occur in a round. We first identify the possible states, and then
derive the steady-state distribution, which facilitates the calculation
of the expected reward.  

For the strategy profile (H, H), both miners compete honestly to generate
the next block. The state transition diagram is demonstrated in Fig.
\ref{fig:queue} (a), and the meaning of each state is shown in Fig.
\ref{fig:states}. Specifically, state 0 represents the state where
the network is consistent, and all the miners do not have a better
solution. According to Fig. \ref{fig:queue} (a), the state transition
probability from state $i$ to state $j$, denoted as $p_{i,j}$,
can be given by

\[
p_{i,j}=\begin{cases}
p_{1}, & (i,j)=(0,2),(1,3),(4,5),(8,9),\\
q_{1}, & i=2,5,6,j=0\text{ or }\\
 & (i,j)=(3,4),(7,8),\\
p_{2}, & (i,j)=(0,1),(2,3),(6,7),\\
q_{2}, & i=1,3,4,7,8,j=0\text{ or }\\
 & (i,j)=(5,6),\\
q_{1}+q_{2}, & i=9,j=0,\\
1-\sum\limits_{\substack{j=0\\j\ne i}}^{9} p_{i,j}, & i=j,\\
0, & \text{otherwise.}
\end{cases}
\]
Practically, the ratio of $s_{1}$ and $s_{2}$may not be a rational
number. The number of possible states may be infinite, since there
always exists a possible case that the opponent's current solution
is still better and remains viable, and the system will not return
to state 0. Here, we make an approximation by truncating these states
and letting state 9 only return to state 0, since the steady-state
probabilities of the states beyond state 9 are negligible. We verify
the effectiveness of the above approximation from the experiments
in Section \ref{sec:Simulations-and-Discussions}. Remark that, in
the special case where $s_{1}:s_{2}=2:3$, miners reach a tie at state
9 with no further states.       

From Fig. \ref{fig:states}, one can observe that, miner 1 can mine
a block with probability $q_{1}$ from state 2, 3, 5, or 9, with optimization
progress $s_{1}$, or from state 6, or 7, with optimization progress
$2s_{1}-s_{2}$. Hence, the expected payoff per round of the miner
1 under the strategy profile (H, H) can be given by
\begin{align}
\pi_{1}\left(\text{H},\text{H}\right) & =q_{1}\left(w_{2}^{H}+w_{3}^{H}+w_{5}^{H}+w_{9}^{H}\right)R\left(s_{1}\right)\label{eq:reward_1h}\\
 & +q_{1}\left(w_{6}^{H}+w_{7}^{H}\right)R\left(2s_{1}-s_{2}\right),\nonumber 
\end{align}
in which $w_{i}^{H}$ represents the steady-state probability for
state $i$ under the strategy profile (H, H). On the other hand, the
expected payoff per round of miner 2 is
\begin{align}
\pi_{2}\left(\text{H},\text{H}\right) & =q_{2}\left(w_{1}^{H}+w_{3}^{H}+w_{7}^{H}\right)R\left(s_{2}\right)\label{eq:reward_2h}\\
 & +q_{2}\left(w_{4}^{H}+w_{5}^{H}\right)R\left(s_{2}-s_{1}\right)\nonumber \\
 & +q_{2}\left(w_{8}^{H}+w_{9}^{H}\right)R\left(2s_{2}-2s_{1}\right).\nonumber 
\end{align}
The expressions of $w_{i}^{H}$ are derived by solving the steady-state
equations from Fig. \ref{fig:queue} (a). To show the solution compactly,
we introduce $v_{i}\triangleq w_{i}^{H}/w_{0}^{H}$ as the relative
value of state $i$ compared with state 0, where the superscript of
$v_{i}$ is dropped since the relative value of every state $i$ is
unique for different strategy profiles. The steady-state probabilities
for other strategy profiles can be derived similarly. We give the
expressions of $v_{i}$ for all the strategy profiles as follows: 

\footnotesize
\begin{align*}
v_{0} & =1,\,\thinspace v_{1}=\frac{p_{2}}{p_{1}+q_{2}}, & v_{2} & =\frac{p_{1}}{q_{1}+p_{2}}, & v_{3} & =\frac{p_{1}v_{1}+p_{2}v_{2}}{q_{1}+q_{2}},\\
v_{4} & =\frac{q_{1}}{p_{1}+q_{2}}v_{3}, & v_{5} & =\frac{p_{1}}{q_{1}+q_{2}}v_{4}, & v_{6} & =\frac{q_{2}}{q_{1}+p_{2}}v_{5},\\
v_{7} & =\frac{p_{2}}{q_{1}+q_{2}}v_{6}, & v_{8} & =\frac{q_{1}}{p_{1}+q_{2}}v_{7}, & v_{9} & =\frac{p_{1}}{q_{1}+q_{2}}v_{8},\\
v_{10} & =\frac{q_{2}}{q_{1}+p_{2}}v_{3}, & v_{11} & =\frac{p_{2}}{q_{1}+q_{2}}v_{10}, & v_{12} & =\frac{q_{1}}{q_{1}+p_{2}}v_{10},\\
v_{13} & =\frac{q_{1}v_{11}+p_{2}v_{12}}{q_{1}+q_{2}}, & v_{14} & =\frac{q_{1}}{p_{1}+q_{2}}v_{3}, & v_{15} & =\frac{q_{2}}{p_{1}+p_{2}}v_{14},\\
v_{16} & =\frac{p_{2}}{p_{1}+q_{2}}v_{15}, & v_{17} & =\frac{p_{1}}{q_{1}+q_{2}}v_{14}, & v_{18} & =\frac{p_{1}v_{15}+q_{2}v_{17}}{q_{1}+p_{2}},\\
v_{19} & =\frac{p_{1}v_{16}+p_{2}v_{18}}{q_{1}+q_{2}}, & v_{20} & =\frac{q_{1}}{p_{1}+q_{2}}v_{19}, & v_{21} & =\frac{p_{1}}{q_{1}+q_{2}}v_{20}.
\end{align*}
\normalsize With $v_{i}$, $w_{i}^{H}$ can be given by
\[
w_{i}^{H}=\frac{v_{i}}{\sum_{j=0}^{9}v_{j}}.
\]

\subsection{Selfish Strategies}

\subsubsection{Fork-and-steal}

Now consider the strategy profile (FS, H), in which miner 1 is selfish.
The states 0-9 are the same as those in Fig. \ref{fig:queue} (a),
while states 10-13 are introduced due to the selfish miner's misbehavior.
The state transition diagram is depicted in Fig. \ref{fig:queue}
(b). When the honest miner generates a block at state 3, instead of
returning to state 0 as honest does, the FS strategy will start an
attack by shifting to state 10. \footnote{When the selfish miner succeeds at state 13, the current solution
of the honest miner is still better and usable. To simplify Markov
chain analysis, we consider a ``dumber'' honest miner by letting
state 13 only return to state 0. Note that it would slightly power
the selfish strategy.} When the selfish miner mines a block at state 12 or 13, the FS strategy
succeeds, and the selfish miner will receive two block rewards, i.e.,
$R\left(s_{1}\right)+R\left(s_{2}-s_{1}\right)$. Hence, the expected
payoff per round of the selfish miner equals
\begin{align}
\pi_{1}\left(\text{FS},\text{H}\right) & =q_{1}\left(w_{2}^{FS}+w_{3}^{FS}+w_{5}^{FS}+w_{9}^{FS}\right)R\left(s_{1}\right)\label{eq:reward_1f}\\
 & +q_{1}\left(w_{6}^{FS}+w_{7}^{FS}\right)R\left(2s_{1}-s_{2}\right)\nonumber \\
 & +q_{1}\left(w_{12}^{FS}+w_{13}^{FS}\right)\left(R\left(s_{2}-s_{1}\right)+R\left(s_{1}\right)\right),\nonumber 
\end{align}
in which $w_{i}^{FS}$ represents the steady-state probability for
state $i$ under the strategy profile (FS, H).

\subsubsection{Ignore-and-fork}

When miner 2 is selfish, consider the strategy profile (H, IF). The
state transitions are shown in Fig. \ref{fig:queue} (c). It can be
viewed that, when the honest miner generates the next block at state
3, the selfish miner will apply the IF strategy and switch to state
14. When the selfish miner fails to generate the second block before
the honest miner at state 19, the optimization result is still valid,
and the system will switch to state 20 instead of state 0. The expected
payoff per round of the selfish miner is
\begin{align}
\pi_{2}\left(\text{H},\text{IF}\right) & =q_{2}\left(w_{1}^{IF}+w_{3}^{IF}+2w_{16}^{IF}+2w_{19}^{IF}\right)R\left(s_{2}\right)\nonumber \\
 & +q_{2}\left(w_{20}^{IF}+w_{21}^{IF}\right)R\left(2s_{2}-2s_{1}\right),\label{eq:reward_2f}
\end{align}
in which $w_{i}^{IF}$ represents the steady-state probability for
state $i$ under the strategy profile (H, IF).      

\subsection{Security Conditions against Selfishness}

Now we would like to derive the security conditions against selfish
miners in our optimization-based PoUW framework. If the strategy profile
(H,H) is a Nash equilibrium, it implies that selfish mining strategies
yield a lower profit, and thus all the miners, even selfish, will
mine honestly for their best interests. Consequently, such security
conditions serve as a foundation for a benign PoUW. We summarize the
conditions as follows. 
\begin{thm}[Security Conditions against Selfishness]
\label{thm:security-conditions} The strategy profile (H, H) is a
Nash equilibrium under the following conditions:
\begin{align}
\alpha_{1}R\left(s_{1}\right)+\beta_{1}R\left(2s_{1}-s_{2}\right)-\gamma_{1}R\left(s_{2}-s_{1}\right) & >0,\label{eq:condv1}\\
\alpha_{2}R\left(s_{2}\right)+\beta_{2}R\left(s_{2}-s_{1}\right)+\gamma_{2}R\left(2s_{2}-2s_{1}\right) & >0,\label{eq:condv2}
\end{align}
in which 
\begin{align*}
\alpha_{1} & =\sum_{i=10}^{13}v_{i}\left(v_{2}+v_{3}+v_{5}+v_{9}\right)-\sum_{i=0}^{9}v_{i}\left(v_{12}+v_{13}\right),\\
\beta_{1} & =\sum_{i=10}^{13}v_{i}\left(v_{6}+v_{7}\right),\quad\gamma_{1}=\sum_{i=0}^{9}v_{i}\left(v_{12}+v_{13}\right),\\
\alpha_{2} & =\left(\sum_{i=14}^{21}v_{i}-\sum_{i=4}^{9}v_{i}\right)\left(v_{1}+v_{3}\right)+\left(\sum_{i=0}^{3}v_{i}+\sum_{i=14}^{21}v_{i}\right)v_{7}\\
 & -2\left(\sum_{i=0}^{3}v_{i}+\sum_{i=4}^{9}v_{i}\right)\left(v_{16}+v_{19}\right),\\
\beta_{2} & =\left(\sum_{i=0}^{3}v_{i}+\sum_{i=14}^{21}v_{i}\right)\left(v_{4}+v_{5}\right),\\
\gamma_{2} & =\left(\sum_{i=0}^{3}v_{i}+\sum_{i=14}^{21}v_{i}\right)\left(v_{8}+v_{9}\right)\\
 & -\left(\sum_{i=0}^{3}v_{i}+\sum_{i=4}^{9}v_{i}\right)\left(v_{20}+v_{21}\right).
\end{align*}
\end{thm}
\begin{proof}
Theorem \ref{thm:security-conditions} is derived by substituting
\eqref{eq:reward_1h}-\eqref{eq:reward_2f} into \eqref{eq:security_cond1}
and \eqref{eq:security_cond2}. 
\end{proof}
Theorem \ref{thm:security-conditions} provides the condition under
which miners tend to comply with the optimization-based PoUW framework.
If \eqref{eq:condv1} and \eqref{eq:condv2} are met, selfish strategies
will yield lower profits than conforming to the framework. Therefore,
a selfish miner is incentivized to be honest, which implies a secure
PoUW. Since $v_{i}$ can be derived recursively, all the coefficients
can be calculated easily. In addition, \eqref{eq:condv1} and \eqref{eq:condv2}
implicitly impose certain requirements on the block reward function
$R\left(s\right)$, which will be elaborated in Section \ref{sec:Reward-Function-Design}. 

\section{Security Overhead Bound \label{sec:PoW-Safeguard-Versus}}

In our framework, PoW not only provides temporal advantages for original
solutions, safeguarding them from being plagiarized, but also ensures
the security of optimization-based PoUW. Nevertheless, it is an additional
overhead that decreases optimization efficiency. This section tries
to qualify the trade-off between security overhead protection and
useful work efficiency. To characterize the ratio of computational
power spent for preserving blockchain security, we introduce the concept
of security overhead ratio, denoted as
\begin{equation}
\eta=\frac{p_{0}}{p_{0}+q_{0}}.\label{eq:eta}
\end{equation}
A smaller $\eta$ indicates that more computational power is utilized
for useful work. However, the security of the blockchain system is
also reduced. Conversely, as $\eta$ grows larger, it offers better
security features. Unfortunately, this leads to less effort spent
on useful work. 

Note that, although Theorem \ref{thm:security-conditions} comprehensively
shows the security conditions, it cannot be used straightforwardly
since it is too complex. Intuitively, the security overhead ratio
$\eta$ is lower bounded to preserve good security properties. Now
we would like to derive a lower bound on the security overhead ratio
$\eta$ via the following Theorem. 
\begin{thm}[Security Overhead Bound]
\label{thm:eta>1/2}  To reach the Nash equilibrium in Theorem \ref{thm:security-conditions},
the security overhead ratio $\eta$ is lower bounded by
\begin{equation}
\eta>\frac{1}{2},\label{eq:1/2}
\end{equation}
i.e., $p_{0}>q_{0}$ must be satisfied. 
\end{thm}
\begin{proof}
When $p_{0}=q_{0}$, we have
\begin{align*}
v_{0} & =1, & v_{1} & =\lambda_{2}, & v_{2} & =\lambda_{1},\\
v_{3} & =2\lambda_{1}\lambda_{2}, & v_{4} & =\lambda_{1}v_{3}, & v_{5} & =\lambda_{1}v_{4},\\
v_{6} & =\lambda_{2}v_{5}, & v_{7} & =\lambda_{2}v_{6}, & v_{8} & =\lambda_{1}v_{7},\\
v_{9} & =\lambda_{1}v_{8}, & v_{10} & =\lambda_{2}v_{3}, & v_{11} & =\lambda_{2}v_{f0},\\
v_{12} & =\lambda_{1}v_{f0}, & v_{13} & =2\lambda_{1}\lambda_{2}v_{f0}.
\end{align*}
Hence, 
\begin{align*}
\alpha_{1} & =\sum_{i=10}^{13}v_{i}\left(v_{2}+v_{3}+v_{5}+v_{9}\right)-\sum_{i=0}^{9}v_{i}\left(v_{12}+v_{13}\right)\\
 & =\left(v_{10}+v_{11}\right)\left(v_{2}+v_{3}+v_{5}+v_{9}\right)\\
 & -\left(v_{12}+v_{13}\right)\left(v_{0}+v_{1}+v_{4}+v_{6}+v_{7}+v_{8}\right)\\
 & <v_{10}\left(1+\lambda_{2}\right)\left(v_{2}+v_{3}+v_{5}+v_{9}\right)\\
 & -v_{10}\left(1+2\lambda_{2}\right)\left(\lambda_{1}v_{0}+\lambda_{1}v_{1}+\lambda_{1}v_{4}+\lambda_{1}v_{8}\right)\\
 & <v_{10}\left(\left(1+\lambda_{2}\right)v_{2}-\left(1+2\lambda_{2}\right)\lambda_{1}v_{0}\right)\\
 & +v_{10}\left(\left(1+\lambda_{2}\right)v_{3}-\lambda_{1}\left(1+2\lambda_{2}\right)v_{1}\right)=0.
\end{align*}
Since $0\leq\lambda_{1},\lambda_{2}\leq1$, $v_{6}=\lambda_{1}^{2}\lambda_{2}v_{3}<v_{3}$,
$v_{7}=\lambda_{2}^{2}v_{5}\leq v_{5}$, we have
\begin{align*}
\beta_{1}-\gamma_{1} & =\sum_{i=10}^{13}v_{i}\left(v_{6}+v_{7}\right)-\sum_{i=0}^{9}v_{i}\left(v_{12}+v_{13}\right)\\
 & <\sum_{i=10}^{13}v_{i}\left(v_{2}+v_{3}+v_{5}+v_{9}\right)-\sum_{i=0}^{9}v_{i}\left(v_{12}+v_{13}\right)\\
 & =\alpha_{1}^{v}<0.
\end{align*}
By setting $s_{2}=\frac{3}{2}s_{1}$, we have
\begin{align*}
\alpha_{1}R\left(s_{1}\right)+\beta_{1}R\left(2s_{1}-s_{2}\right)-\gamma_{1}R\left(s_{2}-s_{1}\right) & =\\
\alpha_{1}R\left(s_{1}\right)+\left(\beta_{1}-\gamma_{1}\right)R\left(\frac{s_{1}}{2}\right) & <0,
\end{align*}
i.e., \eqref{eq:condv1} is not satisfied when $p_{0}=q_{0}$. When
the security overhead ratio $\eta$ grows larger, the consensus mechanism
increasingly resembles pure PoW, which will satisfy the conditions
in Theorem \ref{thm:security-conditions} eventually. Hence, $p_{0}>q_{0}$
is a necessary condition to reach the Nash equilibrium in Theorem
\ref{thm:security-conditions}, and $\eta$ is lower bounded by $\eta>\frac{1}{2}$. 
\end{proof}

Theorem \ref{thm:eta>1/2} shows that, to attain the desired security
conditions in Theorem \ref{thm:security-conditions}, the security
overhead ratio $\eta$ must be greater than 1/2. Failure to meet this
requirement would enable a selfish miner to reap greater profits by
stealing solutions from others, which would prompt more miners to
behave dishonestly, resulting in a breakdown of PoUW. The equivalent
condition, $p_{0}>q_{0}$, implies that the ratio of computational
power allocated for solving optimization problems should be limited.
Recall that, $p_{0}$ is determined by the computational power of
the network and the inherent hardness of the optimization problem,
while $q_{0}$ is determined by both the computational power of the
network and the hardness of PoW. Theorem \ref{thm:eta>1/2} reveals
that, when solving the optimization and PoW are of comparable difficulty,
the ratio of computational power allocated for optimization cannot
exceed that allotted for PoW. One can always meet the constraint by
increasing the difficulty level of PoW, in order to reduce the value
of $q_{0}$. However, this would also lengthen the interval of block
generation, as it becomes harder to mine a valid block. This means,
there is an inherent trade-off between block interval and useful work
efficiency.

Theorem \ref{thm:eta>1/2} also proves that an optimization-based
PoUW cannot work without an anti-plagiarism mechanism, because
in such case $\eta=0$. Hence, to tackle the solution plagiarism issue,
we conjecture that any consensus mechanism based on optimization-solving
should be equipped with countermeasures. Note that, in the framework,
without a direct anti-plagiarism mechanism, PoW serves as
the necessary safeguard, protecting original solutions and preserving
the security of the blockchain. Theorem \ref{thm:eta>1/2} further
suggests that, when optimization and PoW are of comparable hardness,
more than 50\% of the computational power is required as
the security overhead to protect optimization solutions. Of course, embedding the anti-plagiarism countermeasure into the solution may be more efficient, but it may come at the expense of constraining problem types or solution methods.

\section{Reward Principle \label{sec:Reward-Function-Design}}

Security against selfish miners can be attained by proper design of
the reward functions since they always act in their best interests.
Therefore, in this section, we explore the design of the block reward
function $R\left(s\right)$ with respect to the optimization progress
$s$. We first outline its basic properties as below:
\begin{itemize}
\item For any $s>0$, $R\left(s\right)>0$.
\item $R\left(s\right)$ must be monotonically non-decreasing, i.e., $R\left(x_{1}\right)\geq R\left(x_{2}\right)$
holds for any $x_{1}>x_{2}\geq0$. A special case is when $R\left(s\right)$
is a constant, i.e., each block receives the same block reward. 
\end{itemize}

\subsection{Reward Function Design}

Assuming that $R\left(s\right)$ is concave and first-order derivable,
we can obtain a reward function design principle, as shown in the
following theorem.  
\begin{thm}
\label{thm:nece-and-suff-honest-condition}If $\eta>\frac{1}{2}$
and $R\left(s\right)$ is concave and first-order derivable, the security
conditions against selfishness (in Theorem \ref{thm:security-conditions})
are equivalent to the following reward function design principle:
\begin{equation}
R\left(0\right)>\max\left\{ \frac{\gamma_{1}-\alpha_{1}}{\beta_{1}},-\frac{\alpha_{2}}{\beta_{2}+\gamma_{2}}\right\} R\left(s\right).\label{eq:nece-and-suff-condition}
\end{equation}
\end{thm}
To prove Theorem \ref{thm:nece-and-suff-honest-condition}, we first
prove the following lemma about the relative relationship between
$\beta_{2}$ and $\gamma_{2}$.
\begin{lem}
\label{lem:beta-gamma} Under the condition that $p_{0}>q_{0}$, we
have 
\begin{equation}
\beta_{2}>-\frac{16}{5}\gamma_{2}.\label{eq:lem}
\end{equation}
\end{lem}
\begin{proof}
According to the expressions of $v_{i}$, we have 
\[
\sum_{i=14}^{21}v_{i}>\sum_{i=4}^{9}v_{i}.
\]
When $\gamma_{2}\geq0$, since $\beta_{2}>0$, \eqref{eq:lem} obviously
holds. When $\gamma_{2}<0$, we have  \small
\begin{align*}
-\frac{\gamma_{2}}{\beta_{2}} & <\frac{v_{20}+v_{21}-v_{8}-v_{9}}{v_{4}+v_{5}}=\frac{v_{19}-v_{7}}{v_{3}}\\
 & =\frac{\frac{p_{1}}{q_{1}+q_{2}}\left(v_{16}+\frac{p_{2}}{q_{1}+p_{2}}v_{15}\right)}{v_{3}}\\
 & =\frac{p_{1}}{q_{1}+q_{2}}\frac{q_{2}}{p_{1}+p_{2}}\frac{q_{1}}{p_{1}+q_{2}}\left(\frac{p_{2}}{p_{1}+q_{2}}+\frac{p_{2}}{q_{1}+p_{2}}\right)\\
 & =\frac{q_{2}}{q_{1}+q_{2}}\frac{p_{1}}{p_{1}+p_{2}}\left(\frac{p_{2}}{p_{1}+q_{2}}\frac{q_{1}}{p_{1}+q_{2}}+\frac{p_{2}}{q_{1}+p_{2}}\frac{q_{1}}{p_{1}+q_{2}}\right)\\
 & <\frac{q_{1}}{q_{1}+q_{2}}\frac{p_{2}}{p_{1}+p_{2}}\left(\frac{q_{1}p_{2}}{p_{1}q_{2}}+\frac{1}{\left(\frac{1}{q_{1}}+\frac{1}{p_{2}}\right)\left(p_{1}+q_{2}\right)}\right)\\
 & =\lambda_{1}\lambda_{2}\left(1+\frac{1}{\frac{\lambda_{1}}{\lambda_{2}}+\frac{\lambda_{2}}{\lambda_{1}}+\frac{p_{0}}{q_{0}}+\frac{q_{0}}{p_{0}}}\right)\\
 & <\frac{5}{4}\lambda_{1}\lambda_{2}\leq\frac{5}{4}\left(\frac{\lambda_{1}+\lambda_{2}}{2}\right)^{2}=\frac{5}{16},
\end{align*}
\normalsize i.e., $\beta_{2}>-\frac{16}{5}\gamma_{2}$. 
\end{proof}
With Lemma \ref{lem:beta-gamma}, we can then prove Theorem \ref{thm:nece-and-suff-honest-condition}
as below.
\begin{proof}
Define $f\left(s_{1},s_{2}\right)\triangleq\alpha_{1}R\left(s_{1}\right)+\beta_{1}R\left(2s_{1}-s_{2}\right)-\gamma_{1}R\left(s_{2}-s_{1}\right)$
and $g\left(s_{1},s_{2}\right)\triangleq\alpha_{2}R\left(s_{2}\right)+\beta_{2}R\left(s_{2}-s_{1}\right)+\gamma_{2}R\left(2s_{2}-2s_{1}\right)$.
Since $\beta_{1},\gamma_{1}>0$ and $R\left(s\right)$ is monotonically
non-decreasing,
\begin{align*}
\frac{\partial f\left(s_{1},s_{2}\right)}{\partial s_{2}} & =-\beta_{1}R^{\prime}\left(2s_{1}-s_{2}\right)-\gamma_{1}R^{\prime}\left(s_{2}-s_{1}\right)\leq0,
\end{align*}
in which $R^{'}\left(s\right)$ is the first-order derivative of $R\left(s\right)$.
This means, $f\left(s_{1},s_{2}\right)$ is a constant to $s_{2}$
iff $R^{'}\left(s\right)\equiv0$, otherwise $f\left(s_{1},s_{2}\right)$
is monotonically decreasing to $s_{2}$. In either case, we have
\[
f\left(s_{1},s_{2}\right)\geq f\left(s_{1},2s_{1}\right)=\left(\alpha_{1}-\gamma_{1}\right)R\left(s_{1}\right)+\beta_{1}R\left(0\right),
\]
which indicates that, \eqref{eq:condv1} is equivalent to 
\[
\left(\alpha_{1}-\gamma_{1}\right)R\left(s_{1}\right)+\beta_{1}R\left(0\right)>0,
\]
or
\[
R\left(0\right)>\frac{\gamma_{1}-\alpha_{1}}{\beta_{1}}R\left(s_{1}\right).
\]
On the other hand, 
\[
\frac{\partial g\left(s_{1},s_{2}\right)}{\partial s_{1}}=-\beta_{2}R^{'}\left(s_{2}-s_{1}\right)-2\gamma_{2}R^{'}\left(2s_{2}-2s_{1}\right).
\]
If $\gamma_{2}\geq0$, since $\beta_{2}>0$, it is obvious that $\frac{\partial g\left(s_{1},s_{2}\right)}{\partial s_{1}}<0$.
If $\gamma_{2}<0$, from Lemma \ref{lem:beta-gamma} we know that
$\beta_{2}+2\gamma_{2}>0$. When the reward function $R\left(s\right)$
is concave, we have $R^{'}\left(s_{2}-s_{1}\right)\geq R^{'}\left(2s_{2}-2s_{1}\right)\geq0$,
so that
\[
\frac{\partial g\left(s_{1},s_{2}\right)}{\partial s_{1}}<-\left(\beta_{2}+2\gamma_{2}\right)R^{'}\left(s_{2}-s_{1}\right)\leq0.
\]
Hence
\begin{align*}
g\left(s_{1},s_{2}\right) & \geq g\left(s_{2},s_{2}\right)=\alpha_{2}R\left(s_{2}\right)+\left(\beta_{2}+\gamma_{2}\right)R\left(0\right),
\end{align*}
which means, \eqref{eq:condv2} is equivalent to 
\begin{align*}
\alpha_{2}R\left(s_{2}\right)+\left(\beta_{2}+\gamma_{2}\right)R\left(0\right) & >0,
\end{align*}
or
\[
R\left(0\right)>-\frac{\alpha_{2}}{\beta_{2}+\gamma_{2}}R\left(s_{2}\right).
\]
Combining the above analysis, we derive that
\begin{equation}
R\left(0\right)>\max\left\{ \frac{\gamma_{1}-\alpha_{1}}{\beta_{1}},-\frac{\alpha_{2}}{\beta_{2}+\gamma_{2}}\right\} R\left(s\right)\label{eq:reward-eq-concave}
\end{equation}
is equivalent to the conditions in Theorem \ref{thm:security-conditions},
when $\eta>\frac{1}{2}$ and $R\left(s\right)$ is concave. 
\end{proof}

Theorem \ref{thm:nece-and-suff-honest-condition} shows that, if the
reward function $R\left(s\right)$ is concave and first-order derivable,
with $\eta>\frac{1}{2}$ guaranteed (as required by Theorem \ref{thm:eta>1/2}),
to reach security against selfish miners, $R\left(s\right)$ must
satisfy \eqref{eq:nece-and-suff-condition}. With a first-order derivable
concave reward function, as long as \eqref{eq:nece-and-suff-condition}
is satisfied, the selfish miner would gain less by deviating from
honest mining. Note that, this does not mean that the reward function
cannot be non-concave. Even if $R\left(s\right)$ is convex, \eqref{eq:nece-and-suff-condition}
is still a necessary condition, because it is derived by considering
specific pairs of $\left(s_{1},s_{2}\right)$ that minimize the honest-selfish
payoff gap. Nonetheless, additional conditions other than \eqref{eq:nece-and-suff-condition}
may be required.

When $\alpha_{1}\geq\gamma_{1}$ and $\alpha_{2}\geq0$, $\max\left\{ \frac{\gamma_{1}-\alpha_{1}}{\beta_{1}},-\frac{\alpha_{2}}{\beta_{2}+\gamma_{2}}\right\} <0$,
and \eqref{eq:nece-and-suff-condition} is bound to hold. When $\alpha_{1}<\gamma_{1}\text{ or }\alpha_{2}<0,$
$\max\left\{ \frac{\gamma_{1}-\alpha_{1}}{\beta_{1}},-\frac{\alpha_{2}}{\beta_{2}+\gamma_{2}}\right\} >0$,
and because $R\left(s\right)$ is non-decreasing, $\max\left\{ \frac{\gamma_{1}-\alpha_{1}}{\beta_{1}},-\frac{\alpha_{2}}{\beta_{2}+\gamma_{2}}\right\} <1$
must be guaranteed. This means, the following constraints
\begin{align}
\alpha_{1}+\beta_{1}-\gamma_{1} & >0,\label{eq:parameter-necessary1}\\
\alpha_{2}+\beta_{2}+\gamma_{2} & >0,\label{eq:parameter-necessary2}
\end{align}
are to be satisfied, which can be further transformed into a tighter
lower bound on the security overhead ratio $\eta$ than the bound
$\eta>\frac{1}{2}$. When \eqref{eq:parameter-necessary1} is not
satisfied, selfish miners may profit from solution plagiarism, no
matter what reward function is selected. On the other hand, if \eqref{eq:parameter-necessary2}
is not met, selfish miners may deliberately ignore other miners' blocks
for a higher profit, i.e., the chain formation issue. 

With larger value of $\max\left\{ \frac{\gamma_{1}-\alpha_{1}}{\beta_{1}},-\frac{\alpha_{2}}{\beta_{2}+\gamma_{2}}\right\} $,
more constraints are imposed on the reward function design. When $\frac{\gamma_{1}-\alpha_{1}}{\beta_{1}}>-\frac{\alpha_{2}}{\beta_{2}+\gamma_{2}}$,
solution plagiarism is the primary issue faced by the designed framework,
while chain formation becomes more critical if $-\frac{\alpha_{2}}{\beta_{2}+\gamma_{2}}>\frac{\gamma_{1}-\alpha_{1}}{\beta_{1}}$.
Notably, \eqref{eq:parameter-necessary1} and \eqref{eq:parameter-necessary2}
are necessary conditions for \eqref{eq:condv1} and \eqref{eq:condv2},
regardless of the reward function's concavity. Also note that the
necessary condition $\eta>\frac{1}{2}$ is contained in \eqref{eq:parameter-necessary1}
and \eqref{eq:parameter-necessary2}.

Another interpretation of Theorem \ref{thm:nece-and-suff-honest-condition}
is that, if we break the reward into two parts:
\[
R\left(s\right)=R_{0}+\widetilde{R}\left(s\right),
\]
in which $R_{0}=R\left(0\right)$ and $\widetilde{R}\left(0\right)=0$,
we can further transform \eqref{eq:nece-and-suff-condition} into
an upper bound on $R\left(s\right)$, as given by
\begin{equation}
R\left(s\right)<\mu R_{0},\label{eq:upperbound-R}
\end{equation}
where 
\begin{equation}
\mu=\begin{cases}
\frac{1}{\max\left\{ \frac{\gamma_{1}-\alpha_{1}}{\beta_{1}},-\frac{\alpha_{2}}{\beta_{2}+\gamma_{2}}\right\} }, & \alpha_{1}<\gamma_{1}\text{ or }\alpha_{2}<0,\\
+\infty, & \alpha_{1}\geq\gamma_{1}\text{ and }\alpha_{2}\geq0.
\end{cases}\label{eq:mu}
\end{equation}
This indicates that, the block reward must be smaller than an upper
bound proportional to $R_{0}$. $R_{0}$ represents the constant part
of the reward, similar to the block generation reward in Bitcoin.
This means that the extra reward for improving the optimization solution
should constitute a limited share of the total reward. Conversely,
any concave reward function satisfying \eqref{eq:upperbound-R} is
feasible. As a special case, a fixed reward with no extra reward for
optimization progress can always meet the security conditions, as
long as \eqref{eq:parameter-necessary1} and \eqref{eq:parameter-necessary2}
hold. 

In the proof of Theorem \ref{thm:nece-and-suff-honest-condition},
we derive \eqref{eq:parameter-necessary1} when $s_{2}\rightarrow2s_{1}$,
while \eqref{eq:parameter-necessary2} is obtained when $s_{1}\rightarrow s_{2}$.
In a real-life setting, however, we may not desire a solution with
trivial optimization progress to be valid. For instance, when $s_{1}\rightarrow s_{2}$,
if a block with solution $y^{(n)}+s_{1}$ is accepted, a block with
solution $y^{(n)}+s_{2}$ mined afterward would make the optimization
progress $s_{2}-s_{1}$ too trivial. Intuitively, a minimum optimization
progress can be specified to tackle this issue. In that case, \eqref{eq:nece-and-suff-condition}
can be relaxed, as the original boundary conditions ($s_{2}\rightarrow2s_{1}$
or $s_{1}\rightarrow s_{2}$) can no longer be taken.

\subsection{Linear Reward Function}

As a direct application of Theorem \ref{thm:nece-and-suff-honest-condition},
consider a linear reward function 
\begin{equation}
R\left(s\right)=ks+b,s\leq s_{m},\label{eq:linear_reward}
\end{equation}
in which $k\geq0,b>0$, and $s_{m}$ is the maximum optimization improvement
for both miners. A special case is when $k=0$, the block reward (excluding
transaction fees) is a fixed value, as in traditional PoW blockchain
systems. Now we obtain a linear reward design, given by the following
Theorem.
\begin{thm}
\label{thm:linear}If $\eta>\frac{1}{2}$ and the reward function
is linear, a sufficient condition to guarantee security against selfishness
is
\begin{equation}
0\leq k<\frac{\mu-1}{s_{m}}b,\label{eq:linear_constraint}
\end{equation}
in which, $\mu$ is given by \eqref{eq:mu}.
\end{thm}
\begin{proof}
According to Theorem \ref{thm:nece-and-suff-honest-condition}, when
$\eta>\frac{1}{2}$, since the linear reward function is concave and
first-order derivable, \eqref{eq:condv1} and \eqref{eq:condv2} are
equivalent to \eqref{eq:nece-and-suff-condition}. When \eqref{eq:linear_constraint}
is satisfied, we have
\[
ks+b\leq ks_{m}+b<\mu b.
\]
When \eqref{eq:parameter-necessary1} and \eqref{eq:parameter-necessary2}
are met, we have $\max\left\{ \frac{\gamma_{1}-\alpha_{1}}{\beta_{1}},-\frac{\alpha_{2}}{\beta_{2}+\gamma_{2}}\right\} <1$,
hence
\[
b>\max\left\{ \frac{\gamma_{1}-\alpha_{1}}{\beta_{1}},-\frac{\alpha_{2}}{\beta_{2}+\gamma_{2}}\right\} \left(ks+b\right),
\]
i.e., \eqref{eq:nece-and-suff-condition} is satisfied. Therefore,
\eqref{eq:linear_constraint} is a sufficient condition to reach the
desired Nash equilibrium when $\eta>\frac{1}{2}$ and the reward function
is linear.
\end{proof}
Theorem \ref{thm:linear} shows that, the slope of the linear reward
function should be upper bounded. This means that a miner should not
get paid too much for improving the solution. Compared with Theorem
\ref{thm:nece-and-suff-honest-condition}, Theorem \ref{thm:linear}
offers a greater margin to guarantee the security conditions against
selfish miners, making the framework more robust, and its meaning
can be interpreted more obviously. 
\begin{figure}
\begin{raggedright} \centering\includegraphics[width=0.42\textwidth]{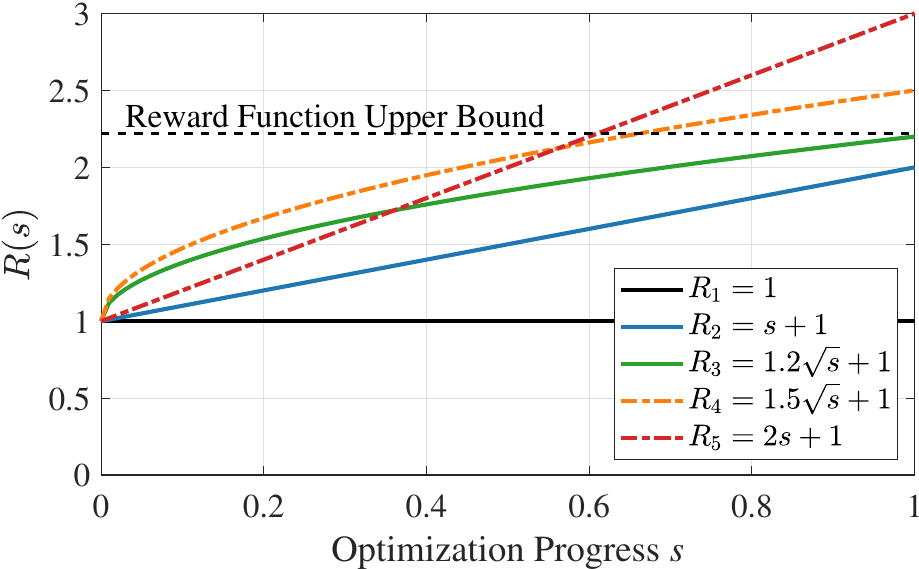}\end{raggedright}
\centering{}\vspace{-0.1cm}

\caption{The example reward functions.\label{fig:reward_functions}}
\vspace{-0.5cm}
\end{figure}

\subsection{Examples}

This subsection demonstrates some design examples: $R_{1}=1$, $R_{2}=s+1$,
$R_{3}=1.2\sqrt{s}+1$, $R_{4}=1.5\sqrt{s}+1$ and $R_{5}=2s+1$,
as shown in Fig. \ref{fig:reward_functions}. Note that, $R_{1}$,
$R_{2}$ and $R_{3}$ satisfy \eqref{eq:nece-and-suff-condition}
while $R_{4}$ and $R_{5}$ do not, as they fail to meet \eqref{eq:upperbound-R}
\footnote{We set $p_{0}=0.005$, $q_{0}=0.001$, $s_{1}=1$ and $s_{2}=1.0001$.
Based on this setting, it can be numerically computed that $\max\left\{ \frac{\gamma_{1}-\alpha_{1}}{\beta_{1}},-\frac{\alpha_{2}}{\beta_{2}+\gamma_{2}}\right\} =0.4412$
and $\mu=2.2663$.}. We show the total reward difference between honest and selfish strategies
in Fig. \ref{fig:simulation4} (see Section \ref{sec:Simulations-and-Discussions}).
One can see that, when the selfish mining power ratio, denoted as
$\lambda_{s}$, approaches 0.5, the differences $\pi_{honest}-\pi_{selfish}$
for $R_{4}$ and $R_{5}$ drop below zero, while the differences for
$R_{1}$, $R_{2}$ and $R_{3}$ always stay above zero. Hence, $R_{1}$,
$R_{2}$ and $R_{3}$ are feasible reward functions in this setting,
while $R_{4}$ and $R_{5}$ are vulnerable to selfish behaviors, which
aligns with Theorem \ref{thm:nece-and-suff-honest-condition}. Intuitively,
a flatter reward function (e.g., a fixed reward) can satisfy Theorem
\ref{thm:nece-and-suff-honest-condition} more easily. However, as
shown in Fig. \ref{fig:secure-region}, a flatter reward function
has higher demands on the security overhead ratio $\eta$ and thus
smaller secure regions, which provides less protection against solution
plagiarism. This indicates a trade-off when tackling solution plagiarism
and chain formation. 

\section{Security Against Maliciousness \label{sec:Long-range-Attack-Analysis}}

\subsection{Necessary Condition against Long-range Attacks}

In addition to selfish miners, malicious miners also pose a threat
to blockchain security and need to be taken into account. Unlike selfish
miners, the malicious miner may attack even at the expense of its
own interests, so the reward function design in Section \ref{sec:Reward-Function-Design}
is powerless against such attacks. In this section, we show that optimization-based
PoUW is secure in the presence of malicious miners, and derive the
security condition against maliciousness.

Consider a powerful malicious miner with a computational power ratio
$\lambda_{1}$, while the computational power ratio of the rest of
the network is $\lambda_{2}=1-\lambda_{1}$. This subsection takes
the long-range attack as an example, in which the goal of the malicious
miner is to generate a longer private chain from $k$ blocks deep.
If the length of the private chain ever exceeds that of the main chain,
the attacker publishes the private chain so that data in the original
$k$ blocks and the following blocks are altered. When chasing the
main chain, the malicious miner can copy the solutions to optimization
problems from existing blocks in the main chain and only needs to
compute PoW, while honest miners need to solve the optimization problem
and the PoW puzzle to generate a block. This makes the malicious miner
more advantageous than in a PoW blockchain. 

Denote the random variables $T_{2}^{o}$ and $T_{2}^{p}$ as the number
of rounds that the honest miner finds a better solution and solves
the PoW puzzle, respectively, and $T_{1}^{p}$ as the number of rounds
that the malicious miner solves PoW. Intuitively, for a secure PoUW
against long-range attacks, the expected time for the honest majority
to generate a block must be smaller than the expected time that the
malicious miner plagiarizes and generates a private block, i.e., 
\begin{equation}
E\left(T_{2}^{o}+T_{2}^{p}\right)<E\left(T_{1}^{p}\right).\label{eq:long-range-expectations}
\end{equation}
Therefore, we immediately derive the following condition against long-range
attacks.
\begin{prop}
\label{prop:Necessary-Long-range-security} A necessary condition
for the framework to guarantee security against long-range attacks
is given by
\begin{equation}
\lambda_{1}<\frac{\eta}{1+\eta},\label{eq:nece-long-range-condition}
\end{equation}
which is also an upper bound on the maximum tolerated adversarial
power.
\end{prop}
\begin{proof}
If \eqref{eq:long-range-expectations} is not satisfied, for any $k$,
the attacker can always catch up and the framework is insecure. Thus
\eqref{eq:long-range-expectations} is a necessary condition. Since
$T_{2}^{o},T_{2}^{p}$ and $T_{1}^{p}$ are independent geometrically
distributed random variables with hit probability $p_{2},q_{2}$ and
$q_{1}$, respectively, their expectations are $\frac{1}{p_{2}},\frac{1}{q_{2}}$
and $\frac{1}{q_{1}}$, respectively. Substituting them into \eqref{eq:long-range-expectations}
yields \eqref{eq:nece-long-range-condition}.
\end{proof}
Proposition \ref{prop:Necessary-Long-range-security} gives a necessary
security condition against long-range attacks. Though \eqref{eq:nece-long-range-condition}
alone is not sufficient to guarantee long-range attack resistance,
the violation of such a condition means the optimization-based PoUW
is vulnerable. Given a security overhead ratio $\eta$, \eqref{eq:nece-long-range-condition}
provides an upper bound for the maximum tolerated adversarial power. 

\subsection{Security Condition against Maliciousness }

Obviously, it is insufficient to consider long-range attacks. To reveal
the security condition against maliciousness, comprehensive malicious
attacks need to be taken into account. Recall that in a secure PoW
blockchain system against malicious attacks, honest miners need to
possess more than half of the computational power. That is, the probability
that honest miners generate the next block before the adversary, denoted
as $P_{h}$, should be greater than $\frac{1}{2}$, and the chain
will have probabilistic convergence to the triumph of the honest chain.
Take the long-range attack for instance. Consider a gambler's ruin
problem on integers. The random walk starts at $k$ (the attack depth)
and the current location corresponds to the main chain's number of
leads, i.e., the number of blocks the malicious miner still needs
to generate in order to catch up with the main chain. In each move,
the walk shifts right with probability $P_{h}$ or left with probability
$1-P_{h}$, and the long-range attack is successful if -1 is ever
reached. From the gambler's ruin model, we know that the long-range
attack success probability, denoted by $P_{lr}\left(k;P_{h}\right)$,
is
\begin{equation}
P_{lr}\left(k;P_{h}\right)=\begin{cases}
\rho^{k+1}, & \rho<1,\\
1, & \rho\geq1,
\end{cases}\label{eq:p_lr_upper_bound}
\end{equation}
in which $\rho=\frac{1-P_{h}}{P_{h}}$. If $\rho<1$, $P_{lr}\left(k;P_{h}\right)$
vanishes exponentially as $k$ increases, which implies good security
features. If $\rho\geq1$, long-range attacks can always destroy the
considered blockchain system. Therefore, $\rho<1$, or equivalently,
$P_{h}>\frac{1}{2}$ is the condition that the long-range attack success
probability drops exponentially with $k$, and therefore guarantees
probabilistic convergence to security against long-range attacks.
Although we use the long-range attack as an example, the above principle
also holds for other malicious attacks such as double-spending. 

Based on the above analysis, we can derive the security condition
against maliciousness in the following theorem.
\begin{thm}[Security Condition against Maliciousness]
\label{thm:security_condition_against_maliciousness} When $\eta>1/2$,
the framework is secure against malicious miners under the following
condition:
\begin{align}
\left(1-q_{1}\right)^{2}\left(1-\frac{q_{1}p_{2}}{p_{2}-q_{2}}\frac{\left(1-q_{2}\right)^{2}}{q_{1}+q_{2}-q_{1}q_{2}}\right)\label{eq:security_condition_against_maliciousness}\\
+\left(1-q_{1}\right)^{2}\frac{q_{1}q_{2}}{p_{2}-q_{2}}\frac{\left(1-p_{2}\right)^{2}}{q_{1}+p_{2}-q_{1}p_{2}} & >\frac{1}{2}.\nonumber 
\end{align}
\end{thm}
\begin{proof}
Consider a PoW blockchain where the honest majority generates a block
before the malicious miner with probability $P\left(T_{2}^{o}+T_{2}^{p}<T_{1}^{p}\right)$.
If $P\left(T_{2}^{o}+T_{2}^{p}<T_{1}^{p}\right)>\frac{1}{2}$, this
PoW blockchain is secure against malicious miners, and so is our framework.
To calculate $P\left(T_{2}^{o}+T_{2}^{p}<T_{1}^{p}\right)$, we first
compute $P\left(T_{2}^{o}+T_{2}^{p}\leq t_{1}-1\right)$, which is
given by
\begin{align}
 & P\left(T_{2}^{o}+T_{2}^{p}\leq t_{1}-1\right)\label{eq:prob_t2o+t2p}\\
 & \overset{(a)}{=}\sum_{t_{2}=1}^{t_{1}-2}P\left(T_{2}^{o}=t_{2}\right)P\left(T_{2}^{p}\leq t_{1}-t_{2}-1\right)\nonumber \\
 & \overset{(b)}{=}\sum_{t_{2}=1}^{t_{1}-2}\left(1-p_{2}\right)^{t_{2}-1}p_{2}\left(1-\left(1-q_{2}\right)^{t_{1}-t_{2}-1}\right)\nonumber \\
 & =1-\frac{p_{2}\left(1-q_{2}\right)^{t_{1}-1}-q_{2}\left(1-p_{2}\right)^{t_{1}-1}}{p_{2}-q_{2}},\nonumber 
\end{align}
where (a) uses the law of total probability, and (b) is because $T_{2}^{o}$
and $T_{2}^{p}$ are geometrically distributed with hit probability
$p_{2}$ and $q_{2}$, respectively. Note that $\eta>1/2$ guarantees
$p_{2}-q_{2}\neq0$. Now we have \small
\begin{align*}
 & P\left(T_{2}^{o}+T_{2}^{p}<T_{1}^{p}\right)\\
 & \overset{(a)}{=}\sum_{t_{1}=3}^{+\infty}P\left(T_{1}^{p}=t_{1}\right)P\left(T_{2}^{o}+T_{2}^{p}\leq t_{1}-1\right)\\
 & \overset{(b)}{=}\sum_{t_{1}=3}^{+\infty}\left(1-q_{1}\right)^{t_{1}-1}q_{1}\left(1-\frac{p_{2}\left(1-q_{2}\right)^{t_{1}-1}-q_{2}\left(1-p_{2}\right)^{t_{1}-1}}{p_{2}-q_{2}}\right)\\
 & =\left(1-q_{1}\right)^{2}\left(1-\frac{q_{1}p_{2}}{p_{2}-q_{2}}\frac{\left(1-q_{2}\right)^{2}}{q_{1}+q_{2}-q_{1}q_{2}}\right)\\
 & +\left(1-q_{1}\right)^{2}\frac{q_{1}q_{2}}{p_{2}-q_{2}}\frac{\left(1-p_{2}\right)^{2}}{q_{1}+p_{2}-q_{1}p_{2}},
\end{align*}
 \normalsize where (a) uses the law of total probability, and (b)
is due to \eqref{eq:prob_t2o+t2p} and the geometric distribution
of $T_{1}^{p}$. Substituting the above result into $P\left(T_{2}^{o}+T_{2}^{p}<T_{1}^{p}\right)>\frac{1}{2}$
yields \eqref{eq:security_condition_against_maliciousness}.
\end{proof}
Theorem \ref{thm:security_condition_against_maliciousness} gives
the condition that guarantees probabilistic convergence to security
against malicious miners. Take long-range attacks for instance. As
long as \eqref{eq:security_condition_against_maliciousness} is satisfied,
the success probability of a long-range attack shrinks exponentially
with the attack depth (pointed out by \eqref{eq:p_lr_upper_bound}),
and our PoUW framework is secure. This is because \eqref{eq:security_condition_against_maliciousness}
transforms the computational power ratio in PoUW into an equivalent
computing power in PoW from the perspective of block generation probability.
Though the malicious miner can gain a certain advantage through solution
plagiarism, optimization-based PoUW is secured if the equivalent honest
computing power remains the majority. It also states that in the optimization-based
PoUW, the maximum tolerated adversarial power ratio is reduced due
to an enhanced attacking strategy (solution plagiarism). Furthermore,
Theorem \ref{thm:security_condition_against_maliciousness} offers
an approach that allows previous security analyses on PoW to be applied
to optimization-based PoUW.  
\begin{figure}
\begin{raggedright} \centering\includegraphics[width=0.43\textwidth]{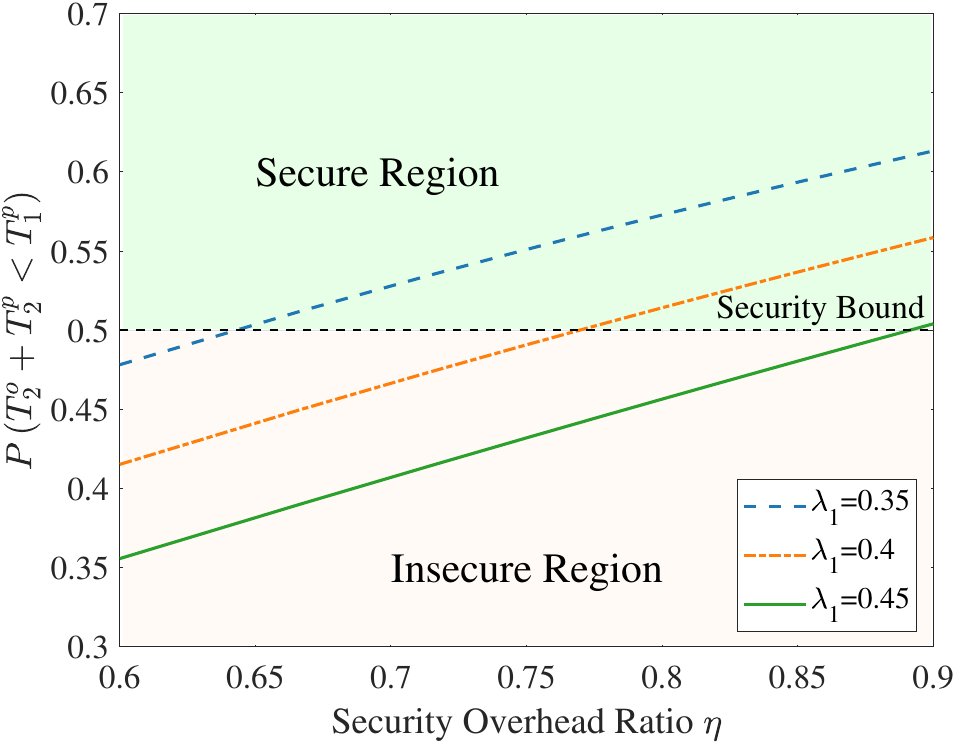}\end{raggedright}
\centering{}\vspace{-0.1cm}

\caption{Security condition against maliciousness in Theorem \ref{thm:security_condition_against_maliciousness} with respect to security
overhead ratio $\eta$. Optimization-based PoUW is secure against
malicious miners if $P\left(T_{2}^{o}+T_{2}^{p}<T_{1}^{p}\right)>\frac{1}{2}$.
\label{fig:long-range}}
\vspace{-0.7cm}
\end{figure}
 
\begin{figure*}
\begin{raggedright} \hfill{}\centering\subfigure[]{ \includegraphics[width=0.43\textwidth]{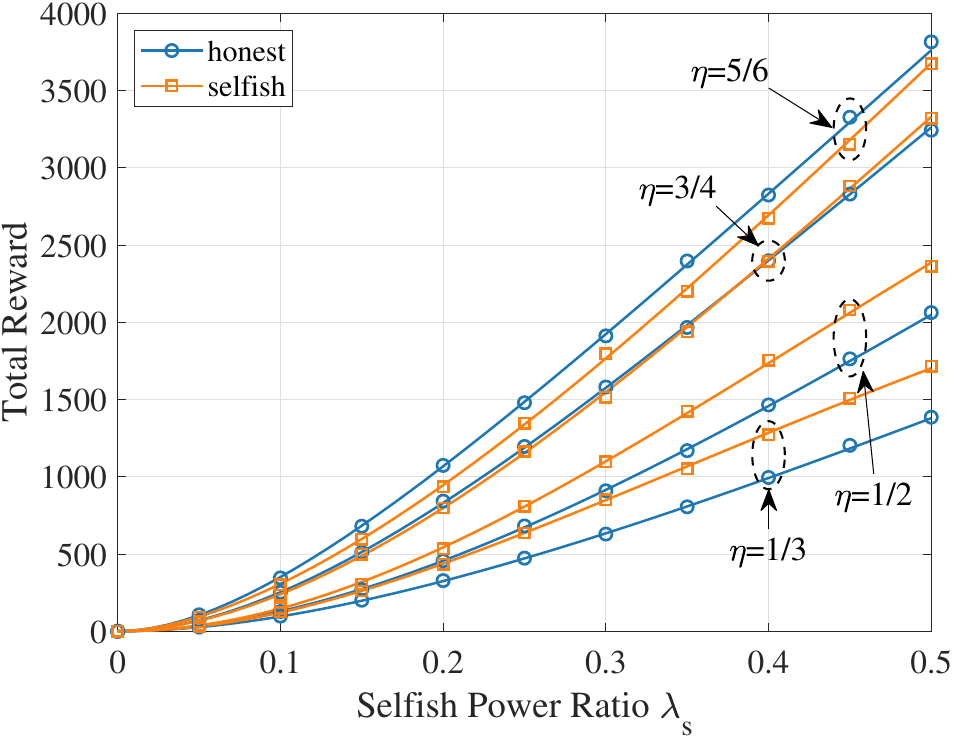}}
\hfill{}\subfigure[]{ \includegraphics[width=0.43\textwidth]{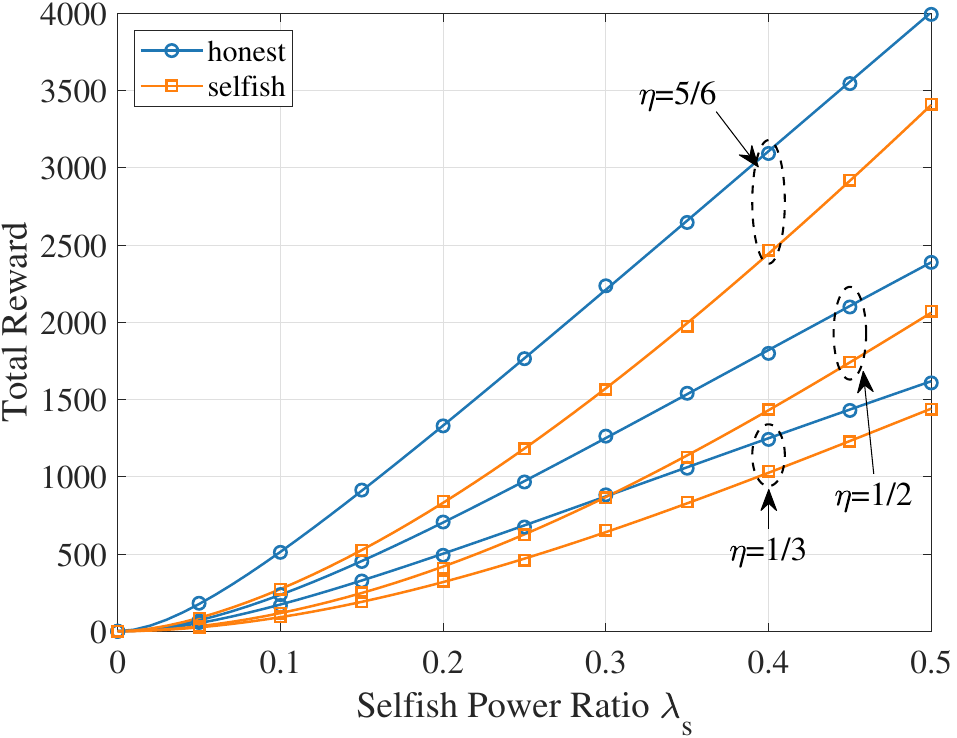}}\hfill{}

\end{raggedright} \centering {}\vspace{-0.3cm}

\caption{Total reward of the selfish miner's honest and selfish strategies
with variant selfish power ratio $\lambda_{s}$. (a) The selfish miner
has a smaller optimization progress. (b) The selfish miner has a larger
optimization progress. \label{fig:simulation1}}
\vspace{-0.3cm}
\end{figure*}

In practical systems, we have $p_{1},p_{2},q_{1},q_{2}\ll1$. In this
case, we can simplify Theorem \ref{thm:security_condition_against_maliciousness}
into the condition regarding the adversarial power ratio $\lambda_{1}$
and the security overhead ratio $\eta$, as shown in the following
corollary.
\begin{cor}
\label{col: approx}When $p_{1},p_{2},q_{1},q_{2}\ll1$, the security
condition in Theorem \ref{thm:security_condition_against_maliciousness}
can be approximated by
\begin{align}
\frac{1}{2}-\frac{\lambda_{1}}{2-\frac{1}{\eta}}+\frac{1}{\left(\frac{\eta}{1-\eta}-1\right)\left(1+\frac{\eta}{1-\eta}\left(\frac{1}{\lambda_{1}}-1\right)\right)} & \geq0.\label{eq:p_approx}
\end{align}
\end{cor}
\begin{proof}
\eqref{eq:p_approx} is derived using $p_{1},p_{2},q_{1},q_{2}\ll1$,
and then substitutes $p_{i}=\lambda_{i}p_{0}$, $q_{i}=\lambda_{i}q_{0}$
and \eqref{eq:eta} into \eqref{eq:security_condition_against_maliciousness}.
\end{proof}
Corollary \ref{col: approx} simplifies Theorem \ref{thm:security_condition_against_maliciousness}
in a bound represented by the adversarial power ratio $\lambda_{1}$
and the security overhead ratio $\eta$. As $\lambda_{1}$ increases,
the bound becomes less easily reached, and a larger $\eta$ is needed
to secure the optimization-based PoUW. Conversely, a larger $\eta$
can prevent a higher proportion of adversarial power from threatening
the blockchain. From numerical results, the approximation in Corollary
\ref{col: approx} is quite close to the accurate bound in Theorem
\ref{thm:security_condition_against_maliciousness}. 

Fig. \ref{fig:long-range} demonstrates $P\left(T_{2}^{o}+T_{2}^{p}<T_{1}^{p}\right)$
with respect to variant security overhead ratio $\eta$. It is observed
that, for a given maximum tolerated adversarial power ratio $\lambda_{1}$,
we can find a minimum required security overhead ratio $\eta$ that
guarantees security against maliciousness. \eqref{eq:security_condition_against_maliciousness}
forms another important bound that limits the security region, which
we will show in Fig. \ref{fig:secure-region}.   
\begin{figure*}
\begin{raggedright}\hfill{} \centering\subfigure[]{ \includegraphics[width=0.43\textwidth]{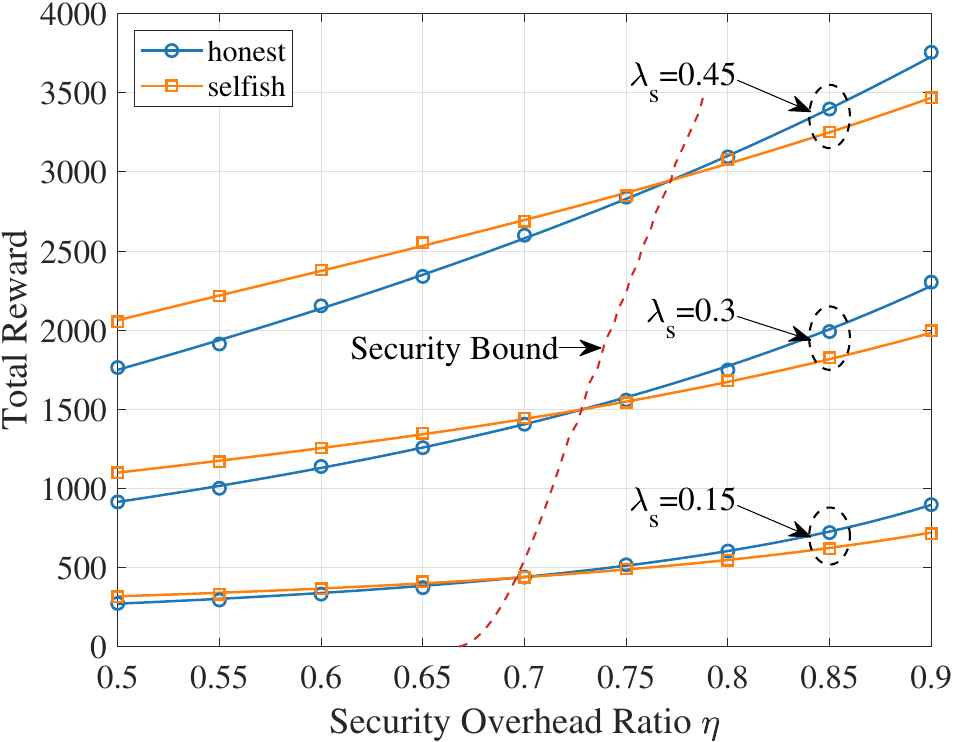}}
\hfill{}\subfigure[]{ \includegraphics[width=0.43\textwidth]{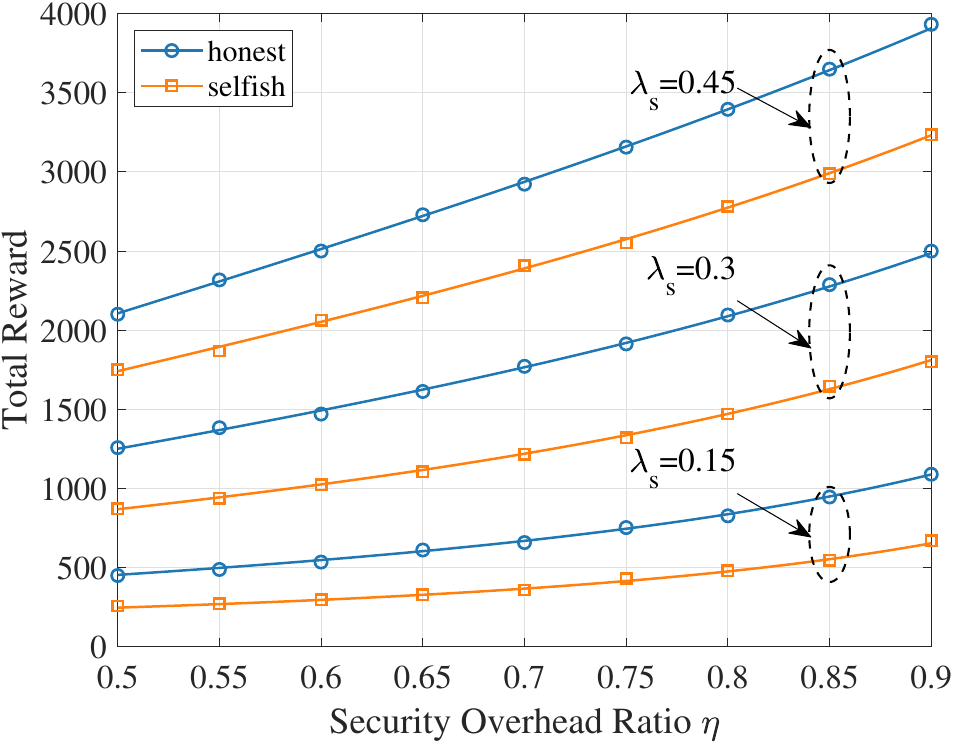}}\hfill{}

\end{raggedright} \centering {}\vspace{-0.3cm}

\caption{Total reward of the selfish miner's honest and selfish strategies
with variant security overhead ratio $\eta$. (a) The selfish miner
has a smaller optimization progress. (b) The selfish miner has a larger
optimization progress. \label{fig:simulation2}}
\vspace{-0.3cm}
\end{figure*}
 
\begin{figure}
\begin{raggedright} \centering\includegraphics[width=0.43\textwidth]{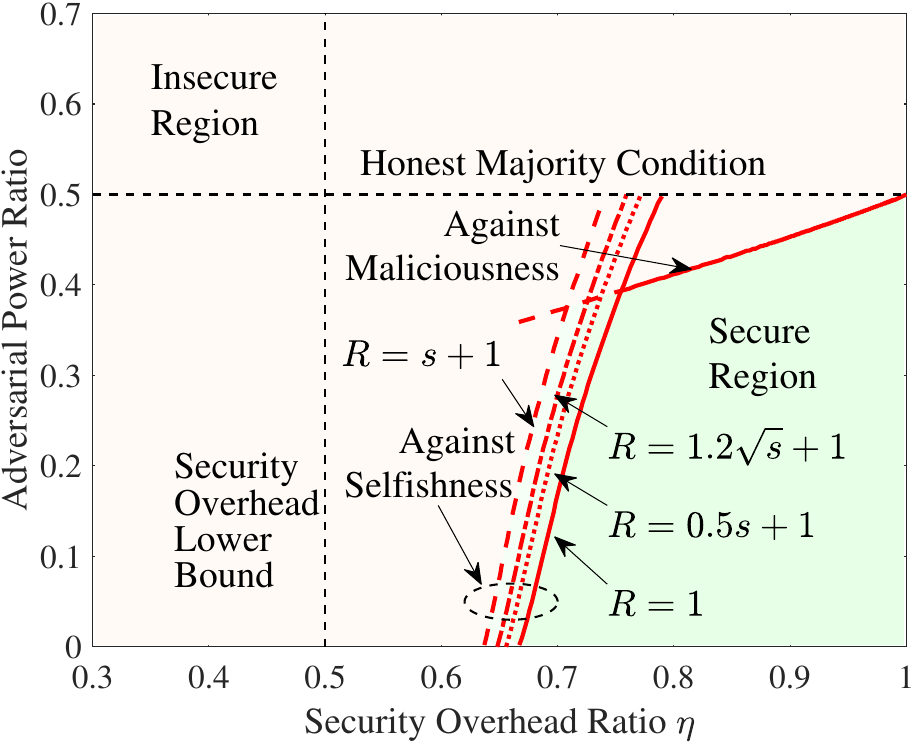}\end{raggedright}
\centering{}\vspace{-0.1cm}

\caption{Illustration of the secure boundaries against selfishness and maliciousness.\label{fig:secure-region}}
\vspace{-0.7cm}
\end{figure}

\section{Simulations and Discussions \label{sec:Simulations-and-Discussions}}
This section presents simulation results to verify and support our
results. For all the figures in this section, lines and markers represent
the analytical and simulation results, respectively. The simulations
are run in a round-based manner by dividing time into fixed-length
intervals \cite{Garay2015}. The experimental parameters are summarized in Table \ref{tab:parameter}. In Fig. \ref{fig:simulation1}, we set $s_{1}:s_{2}=2:3$.
The probability of finding a valid PoW nonce in a round, $q_{0}$,
is set as 0.001, and the probability to find a better solution, $p_{0}$,
is set as 0.0005, 0.001, 0.003, and 0.005, respectively. The block
reward is set as 1, indicating that no extra rewards are offered for
optimization progress. Each of the above system setups is run for
10 million rounds, and we show the total reward of the selfish miner
using honest or selfish strategies under variant selfish computational
power ratio $\lambda_{s}$ in Fig. \ref{fig:simulation1}. We see
that, if the selfish miner has a smaller optimization progress (Fig.
\ref{fig:simulation1}(a)), when the security overhead ratio $\eta$
is low, the selfish strategy (FS) beats honest mining in its total
reward for any $\lambda_{s}\leq0.5$, and the framework is insecure.
As $\eta$ grows larger (see the curves when $\eta=3/4$), the selfish
strategy outperforms honest mining when $\lambda_{s}$ exceeds a certain
value. Finally, when $\eta$ is large enough (see the curves when
$\eta=5/6$), honest mining yields a better payoff for any $\lambda_{s}\leq0.5$,
and the protocol is secure. This verifies Theorem \ref{thm:eta>1/2}
that the framework is only secured when $\eta>\frac{1}{2}$. On the
other hand, when the selfish miner has a larger optimization progress,
with a constant reward function, the total reward for honest mining
is always higher than that of the selfish strategy (IF), regardless
of the value of $\lambda_{s}$ and $\eta$. This indicates that, the
key to preventing the selfish miner from deliberately forking when
it has a larger optimization progress is insensitive to $\lambda_{s}$
and $\eta$, but concerns the reward function design (We will show
that in Fig. \ref{fig:simulation3}). Note that, we do not present
any point when $\lambda_{s}>0.5$, since the blockchain system is
vulnerable to the majority attack. 

\begin{table}[]
    {
    \caption{Experimental parameter table. The top ones are universal parameters.} \label{tab:parameter}
    \centering
    \resizebox{\columnwidth}{!}{
    \begin{tabular}{|c|c|c|}
    \hline
      Symbol   & Meaning & Value\\ \hline 
      \hline
        $s_1$  & Miner 1's optimization progress & 1 \\ \hline

    $q_0$  & Probability of finding a valid nonce per round & 0.001 \\ \hline    
      \hline
      \multicolumn{3}{|c|}{Fig. \ref{fig:simulation1}} \\ \hline

    $p_0$  & Probability of finding a better solution per round & [0.0005,0.005] \\ \hline
    $s_2$  & Miner 2's optimization progress & 1.5 \\ \hline
    $\lambda_s$ & Selfish computing power ratio & [0,0.5] \\ 
    \hline 
    $R$ & Block reward &  1 \\ \hline   \hline
    \multicolumn{3}{|c|}{Fig. \ref{fig:simulation2}} \\ \hline
    $\eta$ & Security overhead ratio & 0.5-0.9\\ \hline
    $p_0$  & Probability of finding a better solution per round & change w.r.t. $\eta$ \\ \hline
    $s_2$  & Miner 2's optimization progress & 1.5 \\ \hline
    $\lambda_s$ & Selfish computing power ratio & 0.15,0.3,0.45 \\ \hline
    $R$ & Block reward &  1 \\ \hline 
    \hline
    \multicolumn{3}{|c|}{Fig. \ref{fig:simulation3}} \\ \hline
    $p_0$  & Probability of finding a better solution per round & 0.005 \\ \hline
    $s_2$  & Miner 2's optimization progress & 1.5 \\ \hline
    $k$ & Block reward slope & 0,2,4,8 \\ \hline
    $R$ & Block reward & $ks+1$ \\ \hline
    \hline
    \multicolumn{3}{|c|}{Fig. \ref{fig:simulation3c}} \\ \hline
    $p_0$  & Probability of finding a better solution per round & 0.005 \\ \hline
    $s_2$  & Miner 2's optimization progress & 1.1,1.5 \\ \hline
    $k$ & Block reward slope & [0,10] \\ \hline
    $R$ & Block reward & $ks+1$ \\ \hline
    \hline
    \multicolumn{3}{|c|}{Fig. \ref{fig:simulation4}} \\ \hline
    $p_0$  & Probability of finding a better solution per round & 0.005 \\ \hline
    $s_2$  & Miner 2's optimization progress & 1.00001 \\ \hline
    $R$ & Block reward & $R_i$ \\ \hline    
    \end{tabular}}
    }
\end{table}

\begin{figure*}
\begin{raggedright} \hfill{}\centering\subfigure[]{ \includegraphics[width=0.43\textwidth]{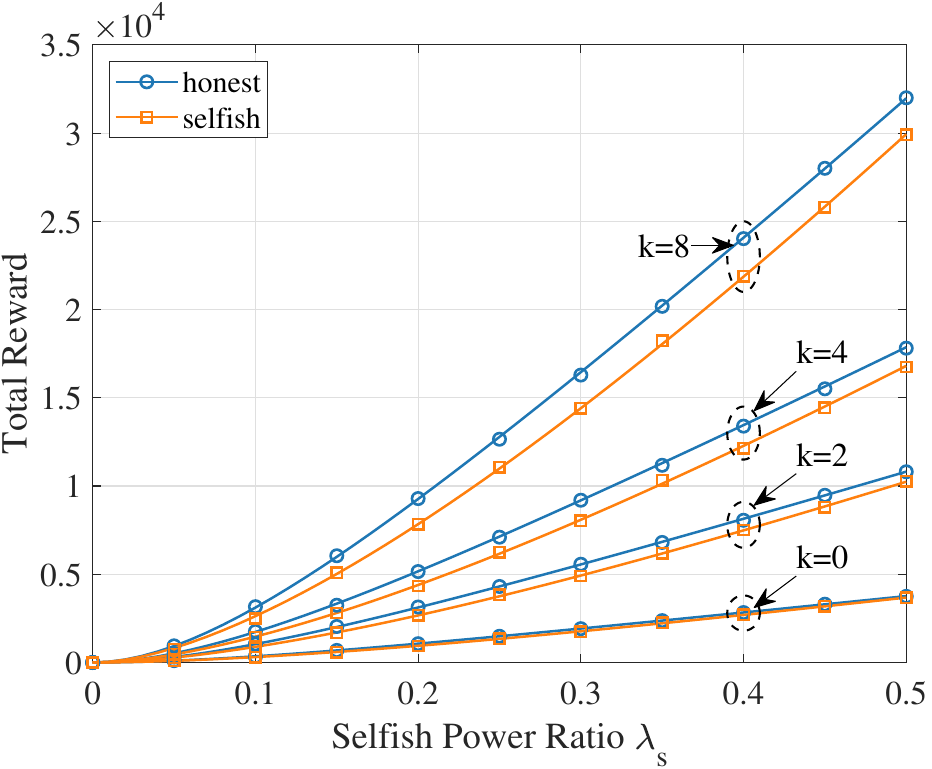}}
\hfill{}\subfigure[]{ \includegraphics[width=0.43\textwidth]{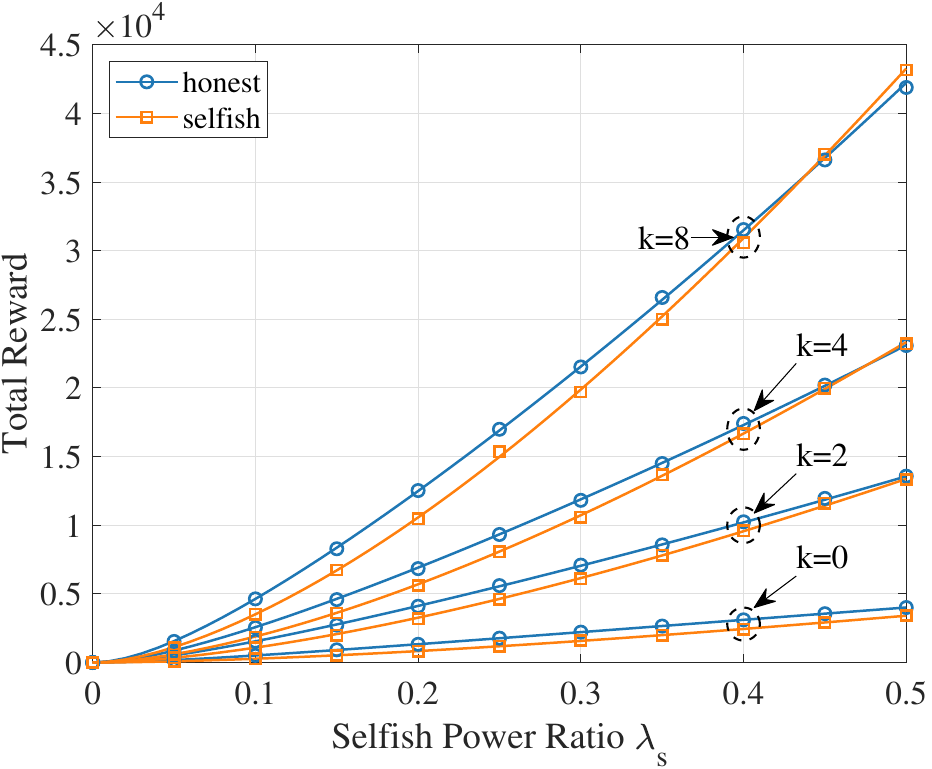}}\hfill{}

\end{raggedright} \centering {}\vspace{-0.3cm}

\caption{Total reward of the selfish miner's honest and selfish strategies
under linear reward function. (a) The selfish miner has a smaller
optimization progress. (b) The selfish miner has a larger optimization
progress. \label{fig:simulation3}}
\vspace{-0.2cm}
\end{figure*}
 
\begin{figure*}
\begin{minipage}[t]{.48\linewidth}

\begin{raggedright} \centering\includegraphics[width=0.9\textwidth]{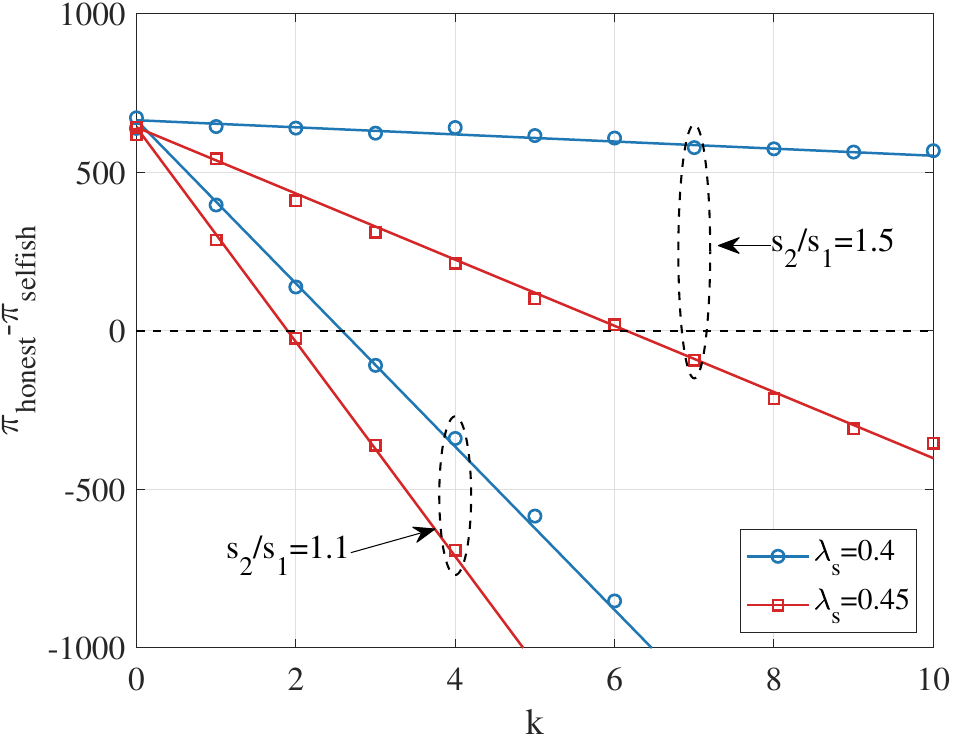}\end{raggedright}
\centering{}\vspace{-0.1cm}

\caption{Honest-selfish reward difference under different linear reward slope
$k$.\label{fig:simulation3c}}
\vspace{-0.3cm} \end{minipage} \hspace{.04\linewidth} \begin{minipage}[t]{.48\linewidth}
\begin{raggedright} \centering\includegraphics[width=0.9\textwidth]{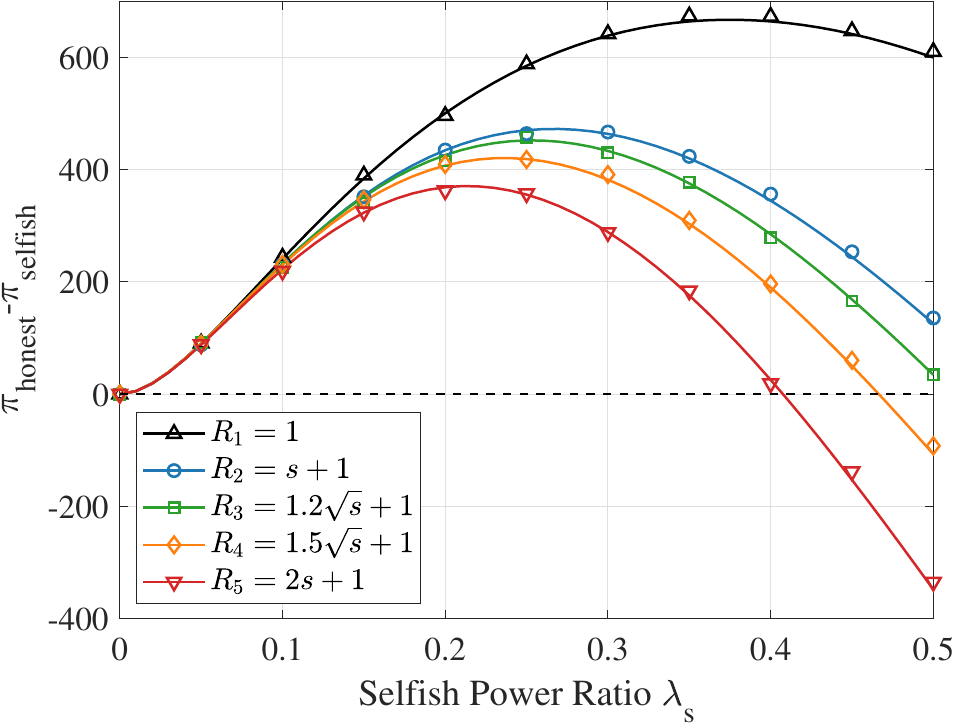}\end{raggedright}
\centering{}\vspace{-0.1cm}

\caption{Honest-selfish reward difference under different selfish power ratio
$\lambda_{s}$.\label{fig:simulation4}}
\vspace{-0.3cm}\end{minipage}
\end{figure*}

In Fig. \ref{fig:simulation2}, we show the impact of the security
overhead ratio $\eta$ on the total rewards of the selfish miner,
for honest and selfish strategies, respectively. The ratio $\eta$
varies from 0.5 to 0.9, and the selfish miner's computational power
ratio is chosen as 0.15, 0.3 and 0.45, respectively. The rest of the
settings are the same as Fig. \ref{fig:simulation1}. From Fig. \ref{fig:simulation2},
one can see that, when the selfish miner has a larger optimization
progress, the total reward of honest mining is always greater than
that of the selfish strategy (IF), while when the selfish miner has
a smaller optimization progress, the total reward of the honest strategy
only surpasses that of the selfish strategy (FS) when the security
overhead ratio $\eta$ exceeds a certain value. As the selfish miner
possesses more computational power, a greater $\eta$ is needed to
secure the framework from solution plagiarism. We mark the security
bound where honest strategy first beats selfish strategy under different
$\lambda_{s}$ in Fig. \ref{fig:simulation2}(a) based on theoretical
results. Fig. \ref{fig:simulation2}(a) not only aligns with Theorem
\ref{thm:eta>1/2}, but also gives a more concrete bound on $\eta$. 

To better illustrate the security bounds, we show the secure region
in Fig. \ref{fig:secure-region} according to the security conditions
against both selfishness and maliciousness. The x-axis represents
the security overhead ratio $\eta$, and the y-axis is the computational
power ratio of the adversary that can be either selfish or malicious.
The horizontal dashed line represents the honest majority assumption,
while the vertical dashed line is the security overhead lower bound
derived in Theorem \ref{thm:eta>1/2}. If the reward function $R=1$
is utilized, for every point located inside the secure region, any
selfish strategies with less than half of the computational power
result in a decrease in the expected payoff. When we use the reward
function $R=0.5s+1$, the secure region expands to the dotted red
line, and when $R=s+1$, the secure region further expands to the
dashed red line. This indicates that, with a greater slope for optimization
progress, the secure region becomes larger, as a smaller $\eta$ can
also secure optimization-based PoUW. Notably, the enlarged secure
region differs for different shapes of the reward function (see the
dash-dot line with $R=1.2\sqrt{s}+1$). Admittedly, the extra reward
for optimization progress is restricted by the upper bound \eqref{eq:upperbound-R}.
Hence, the enlarged area of the secure region is limited. The security
bound against maliciousness can be perceived as a modified honest
majority assumption for a PoUW blockchain that further shrinks the
secure region. A smaller $\eta$ decreases the maximum tolerated malicious
power ratio and leads to a decline in the resistance against malicious
miners.

In Fig. \ref{fig:simulation3} and Fig. \ref{fig:simulation3c}, we
apply a linear reward function $R\left(s\right)=ks+1$ with different
slope $k$, and demonstrate the total reward of honest and selfish
strategies under variant selfish power ratio $\lambda_{s}$. In Fig. \ref{fig:simulation3}, we set
$s_{1}=1$ and $s_{2}=1.5$, and $p_{0}$ and $q_{0}$ are set as
0.005 and 0.001, respectively. Fig. \ref{fig:simulation3c} uses the same parameters except that $s_2=1.1$ or $1.5$.  In Fig. \ref{fig:simulation3}(a),
it can be viewed that for different $k$, the honest strategy always
yields a higher total reward than the selfish strategy (FS), since
$\eta=5/6$ is large enough. In Fig. \ref{fig:simulation3}(b), as
$k$ increases, the amount of extra reward for solution improvement
grows, and the total reward of the selfish strategy (IF) surpasses
that of honest mining when $\lambda_{s}$ is large, making optimization-based
PoUW insecure. Fig. \ref{fig:simulation3c} further shows the total
reward difference between honest and selfish strategies, $\pi_{honest}-\pi_{selfish}$,
with different optimization progress, when the selfish miner has a
larger optimization progress. We see that, the difference decreases
as the slope $k$ increases. For smaller $s_{2}/s_{1}$, the difference
drops below zero more rapidly, meaning that the bound on the reward
function is tighter. This agrees with the proof of Theorem \ref{thm:nece-and-suff-honest-condition}
that the constraint \eqref{eq:condv2} becomes the tightest when $s_{2}\rightarrow s_{1}$. 

With $s_{1}=1$ and $s_{2}=1.0001$, we demonstrate the total reward
difference between honest and selfish strategies, with variant reward
functions, as shown in Fig. \ref{fig:simulation4}. It is observed
that, the difference $\pi_{honest}-\pi_{selfish}$ first rises from
zero (since when $\lambda_{s}=0$, we have $\pi_{honest}=\pi_{selfish}=0$),
then drops as the selfish power ratio $\lambda_{s}$ increases. The
reward functions $R_{1}=1$, $R_{2}=s+1$ and $R_{3}=1.2\sqrt{s}+1$
guarantee that $\pi_{honest}-\pi_{selfish}>0$ for any $0<\lambda_{s}<0.5$,
which preserves good security properties. In contrast, with the reward
functions $R_{4}=1.5\sqrt{s}+1$ and $R_{5}=2s+1$, a selfish miner
with $\lambda_{s}$ large enough can profit more from selfish strategies,
thus threatening the security of optimization-based PoUW. Fig. \ref{fig:simulation4}
provides a more concrete vision of the reward function design principle
that can be applied to actual blockchain systems. 

\section{Conclusions \label{sec:Conclusions} }

In this paper, we formalized a generic optimization-based PoUW framework,
in which the computational power wasted for useless PoW puzzles can
be partially reinvested into solving optimization problems. We analyzed our framework
by modeling the network and different strategies properly. For selfish
miners, we obtained the security conditions against selfishness, and derived the equivalent reward principle explicitly. Furthermore,
we looked into the trade-off between optimization efficiency and PoW
safeguard, and uncovered the lower bound for the security overhead.
We also considered our framework in the presence of malicious miners
and derived the security condition against maliciousness. Simulation
results validated the effectiveness of the framework together with
the correctness of the theoretical analysis.

In general, our analysis yields a relatively conservative conclusion: optimization-based PoUWs without direct plagiarism defense require necessary security overhead to maintain blockchain security. Future research may look into problem-specific anti-plagiarism methods to reduce the security overhead. Note that PoW might be replaced by some lightweight cryptographic primitives, such as verifiable delay functions (VDFs). Since the VDF evaluation is sequential and non-parallelizable, a fixed minimum latency is induced before block publication, which fits the security overhead. However, any adaptation that tries to replace the computationally heavy security overhead should be designed carefully to prevent potential vulnerabilities such as grinding. Another source of inefficiency comes from the fact that, despite being useful, all miners are computing the same optimization problem, which is essentially a waste of computation. The way that miners solve optimization problems may not have to replicate such a competing nature. For future research, our optimization-based PoUW framework can be extended by considering lightweight security overhead substitutes, processing parallelism, and throughput enhancement.

\appendices{}

\bibliographystyle{IEEEtran}
\bibliography{reference}

\begin{IEEEbiography}[{\includegraphics[width=1in,height=1.25in,clip,keepaspectratio] {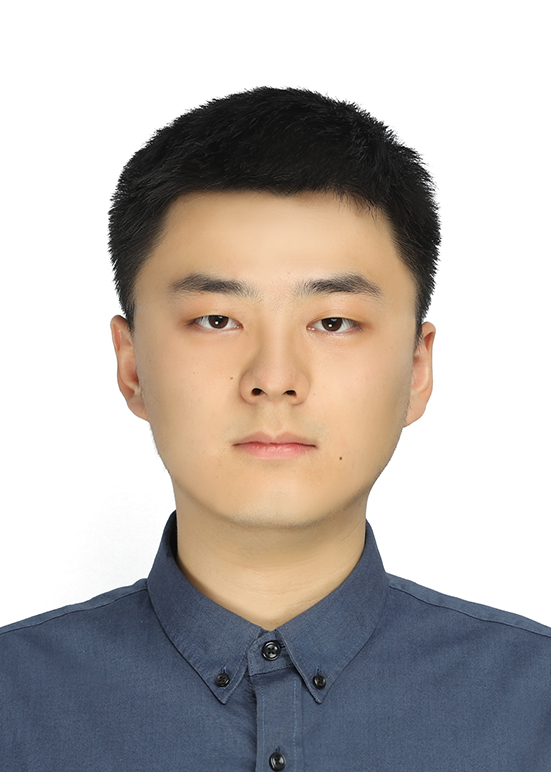}}]{Weihang Cao}
(Graduate Student Member, IEEE) received the B.E. and M.S. degree in electrical engineering from Southeast University, Nanjing, China, in 2021 and
2024, respectively. His research interests include blockchain technologies, optimization and wireless networks. \end{IEEEbiography}

\vspace{-3em}

\begin{IEEEbiography}[{\includegraphics[width=1in,height=1.25in,clip,keepaspectratio] {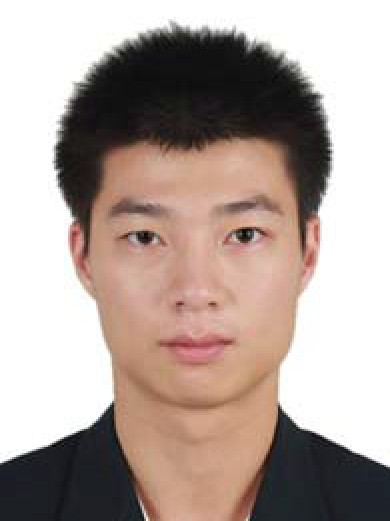}}]{Xintong Ling}
(Member, IEEE) received the B.E. and
Ph.D. degrees in electrical engineering from Southeast University,
Nanjing, China, in 2013 and 2018, respectively. From 2016 to 2018,
he was a visiting Ph.D. student at the Department of Electrical and
Computer Engineering, University of California at Davis. He is currently
an Associate Professor with the National Mobile Communications Research
Laboratory (NCRL), Southeast University, and also with the Purple
Mountain Laboratories. His current research interests focus on future-generation
wireless communications and networks, including blockchain technologies,
distributed systems, optical wireless communications, network modeling,
machine learning, and signal processing. He was selected for the Young
Elite Scientists Sponsorship Program by the China Association for
Science and Technology. \end{IEEEbiography}

\vspace{-3em}

\begin{IEEEbiography}[{\includegraphics[width=1in,height=1.25in,clip,keepaspectratio] {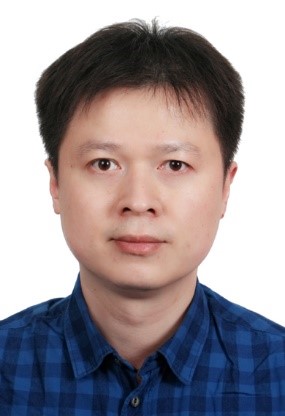}}]{Jiaheng Wang}
(Senior Member, IEEE) received the Ph.D. degree
in electronic and computer engineering from the Hong Kong University
of Science and Technology, Kowloon, Hong Kong, in 2010, and the B.E.
and M.S. degrees from the Southeast University, Nanjing, China, in
2001 and 2006, respectively.

He is currently a Full Professor at the National Mobile Communications
Research Laboratory (NCRL), Southeast University, Nanjing, China.
He is also with Purple Mountain Laboratories Nanjing, China. From
2010 to 2011, he was with the Signal Processing Laboratory, KTH Royal
Institute of Technology, Stockholm, Sweden. He also held visiting
positions at the Friedrich Alexander University Erlangen-Nürnberg,
Nürnberg, Germany, and the University of Macau, Macau. His research
interests are mainly on communication systems, wireless networks and
blockchain technologies.

Dr. Wang has published more than 170 articles on international journals
and conferences. He serves as an Associate Editor for the IEEE Transactions
on Wireless Communications, an Associate Editor for the IEEE Transactions
on Communications, and a Senior Area Editor for the IEEE Signal Processing
Letters. He was a recipient of the Humboldt Fellowship for Experienced
Researchers and the best paper awards of IEEE GLOBECOM 2019, ADHOCNETS
2019, and WCSP 2022 and 2014. \end{IEEEbiography}

\vspace{-3em}

\begin{IEEEbiography}[{\includegraphics[width=1in,height=1.25in,clip,keepaspectratio] {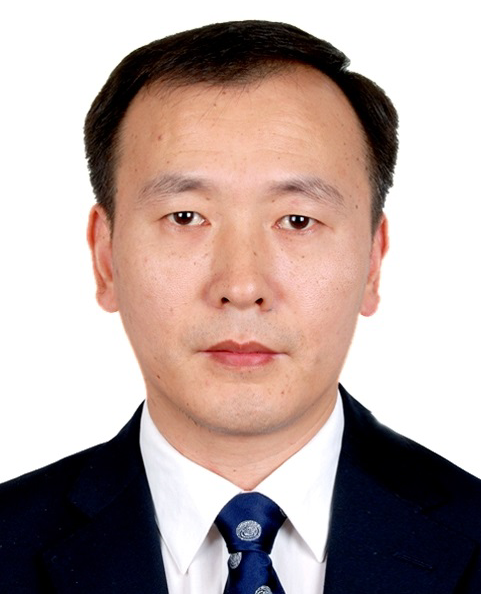}}]{Xiqi Gao}
(Fellow, IEEE)
received the Ph.D. degree in electrical engineering from Southeast
University, Nanjing, China, in 1997.

He joined the Department of Radio Engineering, Southeast University,
in April 1992. From September 1999 to August 2000, he was a Visiting
Scholar with the Massachusetts Institute of Technology, Cambridge,
MA, USA, and Boston University, Boston, MA, USA. Since May 2001, he
has been a Professor of information systems and communications. From
August 2007 to July 2008, he visited the Darmstadt University of Technology,
Darmstadt, Germany, as a Humboldt Scholar. His current research interests
include broadband multicarrier communications, massive MIMO wireless
communications, satellite communications, optical wireless communications,
information theory, and signal processing for wireless communications.

Dr. Gao received the Science and Technology Awards of the State Education
Ministry of China in 1998, 2006, and 2009, the National Technological
Invention Award of China in 2011, the Science and Technology Award
of Jiangsu Province of China in 2014, and the 2011 IEEE Communications
Society Stephen O. Rice Prize Paper Award in the field of communications
theory. From 2007 to 2012, he served as an Editor for the IEEE Transactions
on Wireless Communications. From 2009 to 2013, he served as an Associate
Editor for the IEEE Transactions on Signal Processing. From 2015 to
2017, he served as an Editor for the IEEE Transactions on Communications.
\end{IEEEbiography}

\vspace{-3em}

\begin{IEEEbiography}[{\includegraphics[width=1in,height=1.25in,clip,keepaspectratio] {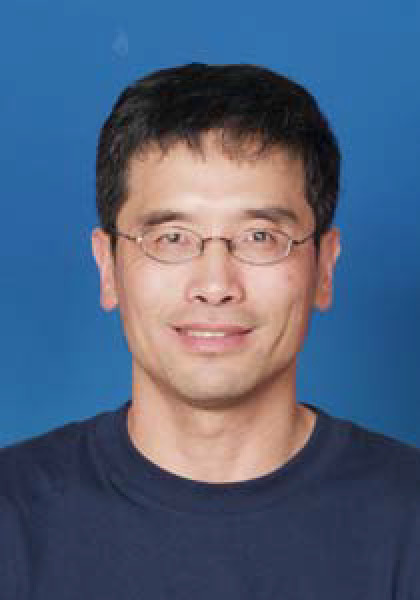}}]{Zhi Ding}
(Fellow, IEEE)
received the Ph.D. degree in electrical engineering from Cornell University
in 1990.,He is currently with the Department of Electrical and Computer
Engineering, University of California at Davis (UC Davis), Davis,
where he was a Distinguished Professor. From 1990 to 2000, he was
a Faculty Member of Auburn University and later the University of
Iowa. He joined the College of Engineering, UC Davis, in 2000. His
major research interests and expertise cover the areas of wireless
networking, communications, signal processing, multimedia, and learning.
He has supervised over 30 Ph.D. dissertations, since joining UC Davis.
His research team of enthusiastic researchers works very closely with
industry to solve practical problems and contributes to technological
advances. His team has collaborated with researchers around the world
and welcomes self-motivated young talents as new members. He currently
serves as the Chief Information Officer and the Chief Marketing Officer
of the IEEE Communications Society. He was an Associate Editor of
IEEE Transactions on Signal Processing from 1994 to 1997 and from
2001 to 2004 and an Associate Editor of IEEE Signal Processing Letters
from 2002 to 2005. He was a member of Technical Committee on Statistical
Signal and Array Processing and a member of Technical Committee on
Signal Processing for Communications (1994--2003). He was the General
Chair of the 2016 IEEE International Conference on Acoustics, Speech,
and Signal Processing and the Technical Program Chair of the 2006
IEEE Globecom. He was also an IEEE Distinguished Lecturer of Circuits
and Systems Society (2004--2006) and Communications Society (2008--2009).
He served as a Steering Committee Member for IEEE Transactions on
Wireless Communications (2007--2009) and the Chair (2009--2010).
He is the coauthor of the textbook: Modern Digital and Analog Communication
Systems (5th edition, Oxford University Press, 2019). He received
the IEEE Communication Society\textquoteright s WTC Award in 2012
and the IEEE Communication Society\textquoteright s Education Award
in 2020. \end{IEEEbiography}
\end{document}